\newif\iflipics

\iflipics

\documentclass[a4paper,UKenglish,cleveref, autoref, thm-restate]{lipics-v2021}

\title{SETH-based Lower Bound for Dynamic Degeneracy} 
\titlerunning{SETH-based Lower Bound for Dynamic Degeneracy} 

\author{Konrad Majewski}{Institute of Informatics, University of Warsaw, Poland}{k.majewski@mimuw.edu.pl}{}{}
\author{Micha\l{} Pilipczuk}{Institute of Informatics, University of Warsaw, Poland}{michal.pilipczuk@mimuw.edu.pl}{}{}

\authorrunning{K. Majewski, Mi. Pilipczuk} 

\Copyright{Konrad Majewski, Micha\l{} Pilipczuk} 



\funding{\flag[0.17\textwidth]{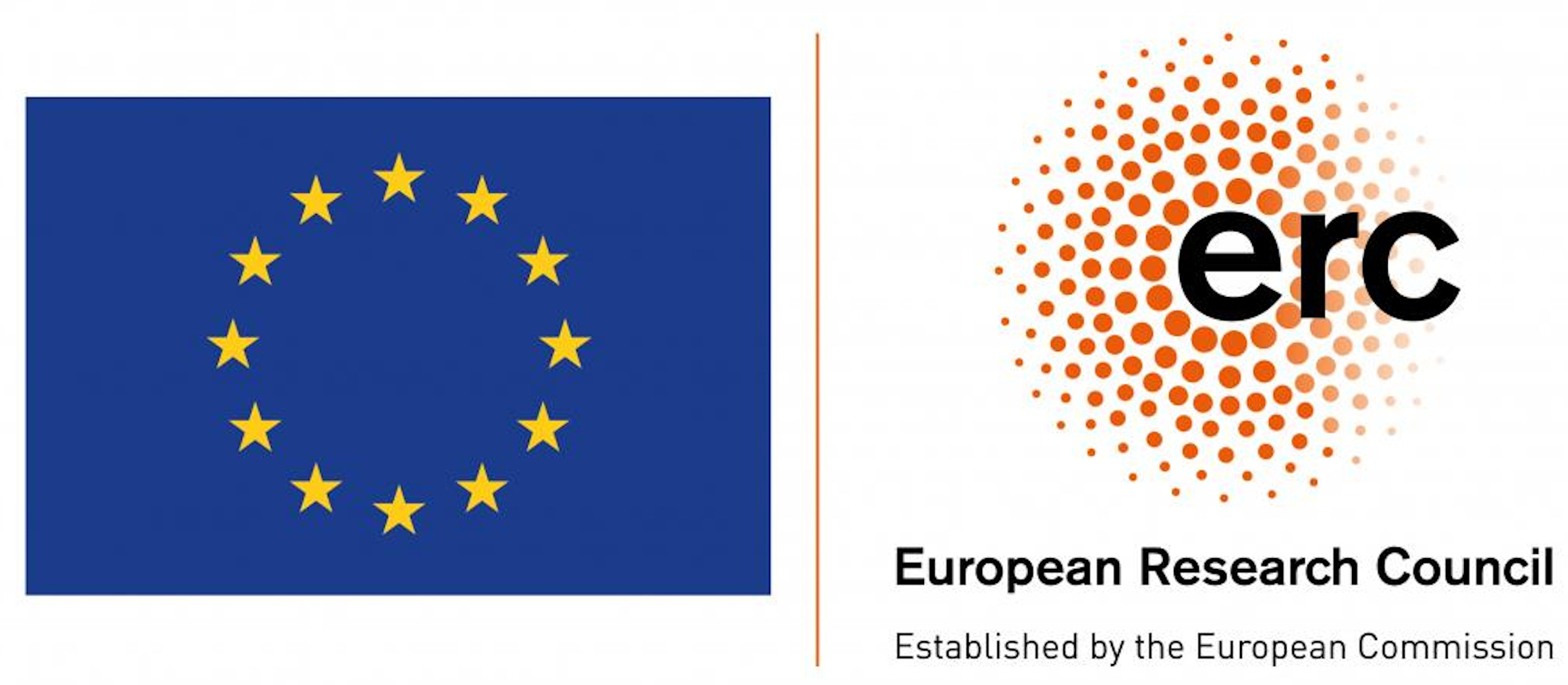}The work of both authors is a part of project BOBR that has received funding from the European Research Council (ERC) under the European Union’s Horizon 2020 research and innovation programme (grant agreement No. 948057).}

\EventEditors{John Q. Open and Joan R. Access}
\EventNoEds{2}
\EventLongTitle{42nd Congress of Visibly Introvert Theoreticians (CVIT 2016)}
\EventShortTitle{CVIT 2016}
\EventAcronym{CVIT}
\EventYear{2016}
\EventDate{December 24--27, 2016}
\EventLocation{Little Whinging, United Kingdom}
\EventLogo{}
\SeriesVolume{42}
\ArticleNo{23}

\else 

\documentclass[11pt]{article}
\usepackage{a4wide}
\usepackage[utf8]{inputenc}
\usepackage[T1]{fontenc}

\title{SETH-based Lower Bound for Dynamic Degeneracy\thanks{The work of both authors is a part of project BOBR that has received funding from the European Research Council (ERC) under the European Union’s Horizon 2020 research and innovation programme (grant agreement No. 948057).}
}
\author{Konrad Majewski\thanks{Institute of Informatics, University of Warsaw, Poland (\texttt{k.majewski@mimuw.edu.pl})} \and
Michał Pilipczuk\thanks{Institute of Informatics, University of Warsaw, Poland (\texttt{michal.pilipczuk@mimuw.edu.pl})}}
\date{}

\fi

\usepackage{inconsolata}
\usepackage{libertine}

\usepackage{amsthm}
\usepackage{amssymb}
\usepackage{amsmath}
\usepackage{bbm} 
\usepackage{mathtools}
\usepackage{mathrsfs}
\usepackage{stmaryrd} 
\usepackage{thmtools,thm-restate}

\usepackage{tikz}
\usetikzlibrary{fit, positioning, calc, backgrounds}
\usepackage{graphicx}
\graphicspath{ {./images/} }

\usepackage{algorithm}
\usepackage[noend]{algpseudocode}
\usepackage{caption}
\usepackage{subcaption}
\usepackage[noadjust]{cite}
\usepackage{tcolorbox}
\usepackage{textpos}
\usepackage{todonotes}
\usepackage{xcolor}
\usepackage{xfrac}
\usepackage{xparse}
\usepackage{xspace}

\iflipics
\else 
\usepackage[colorlinks = true,linkcolor = blue,urlcolor = blue, citecolor = blue]{hyperref}
\usepackage[capitalize, nameinlink]{cleveref}
\fi

\iflipics
\else 
\newtheorem{theorem}{Theorem}[section]

\newtheorem{lemma}[theorem]{Lemma}

\crefname{claim}{Claim}{Claims}
\newtheorem{claim}[theorem]{Claim}
\newcommand{\cqed}{\ensuremath{\lhd}}
\makeatletter
\newenvironment{claimproof}{\par
	\pushQED{\cqed}%
	\normalfont \topsep6\p@\@plus6\p@\relax
	\trivlist
	\item\relax
	{\itshape
		Proof of the claim\@addpunct{.}}\hspace\labelsep\ignorespaces
}{%
	\hfill\popQED\endtrivlist\@endpefalse
}
\makeatother
\fi

\newtheorem{fact}[theorem]{Fact}

\newcommand{\Inline}[1]{#1\xspace}

\newcommand{\ceil}[1]{\left\lceil #1 \right\rceil}

\newcommand{\ProblemName}[1]{\Inline{\textsc{#1}}}
\newcommand{\ov}{\ProblemName{Orthogonal Vectors}}
\newcommand{\dyncir}{\ProblemName{Dynamic Circuit Evaluation}}

\newcommand{\dyndegany}[2]{\ProblemName{Dynamic \ensuremath{(#1, #2)}-Approximate Degeneracy}}
\newcommand{\dyndeg}[2]{\ProblemName{Dynamic \ensuremath{(2-#1, #2)}-Approximate Degeneracy}}

\newcommand{\yesinstance}{\Inline{\textsc{Yes}-instance}}
\newcommand{\noinstance}{\Inline{\textsc{No}-instance}}

\newcommand{\oper}[1]{\mathtt{#1}}

\newcommand{\Oh}{\mathcal{O}}

\newcommand{\N}{\mathbb{N}}

\newcommand{\Z}{\mathbb{Z}}

\newcommand{\eps}{\varepsilon}

\renewcommand{\leq}{\leqslant}
\renewcommand{\geq}{\geqslant}

\renewcommand{\le}{\leqslant}
\renewcommand{\ge}{\geqslant}

\newcommand{\scalar}[2]{\langle #1, #2 \rangle}

\DeclareMathOperator{\arb}{arb}
\DeclareMathOperator{\mad}{mad}
\DeclareMathOperator{\ort}{ort}

\newcommand{\mx}{\mathrm{max}}

\newcommand{\dsfont}[1]{\ensuremath{\mathsf{#1}}}
\newcommand{\Dcirc}{\dsfont{DCircuit}}
\newcommand{\Dovtocirc}{\dsfont{OvToCircuit}}
\newcommand{\Dcirctodeg}{\dsfont{CircuitToDeg}}
\newcommand{\Ddeg}{\dsfont{DDegeneracy}}

\newcommand{\cinputs}{\oper{Inputs}}

\newcommand{\cgates}{\oper{Gates}}
\newcommand{\cwires}{\oper{Wires}}

\newcommand{\orgates}{\oper{Gates}^{\textsc{OR}}}
\newcommand{\andgates}{\oper{Gates}^{\textsc{AND}}}

\newcommand{\maxdeg}{3\ell}

\begin{document}
\maketitle

\begin{abstract}
In this work, we consider the problem of maintaining an approximate value of degeneracy of a given dynamic $n$-vertex graph $G$ updated by edge insertions and deletions.
From the work of Christiansen and Rotenberg [ICALP 2022], it follows that one can design a dynamic data structure for this problem with worst-case update time $\text{poly}(d_{\mathrm{max}}, \log n)$ that maintains an integer between $d$ and $2d+3$ where $d$ is the degeneracy of $G$, under the assumption that $d$ never exceeds~$d_{\mathrm{max}}$.
We complement their result by providing a conditional lower bound:
we prove that, unless SETH fails, for any $\varepsilon, \delta > 0$, $k \in \mathbb{N}$, and function $f\colon \mathbb{N}\to \mathbb{N}$, there is no data structure which maintains a $(2-\varepsilon)$-approximation of the degeneracy of $G$ with initialization time $f(d_{\mathrm{max}})\cdot n^k$ and amortized update time $f(d_{\mathrm{max}})\cdot n^{1-\delta}$.
\end{abstract}

\iflipics
\else 
\thispagestyle{empty}
\begin{textblock}{20}(-1.9, 6.2)
    \includegraphics[width=40px]{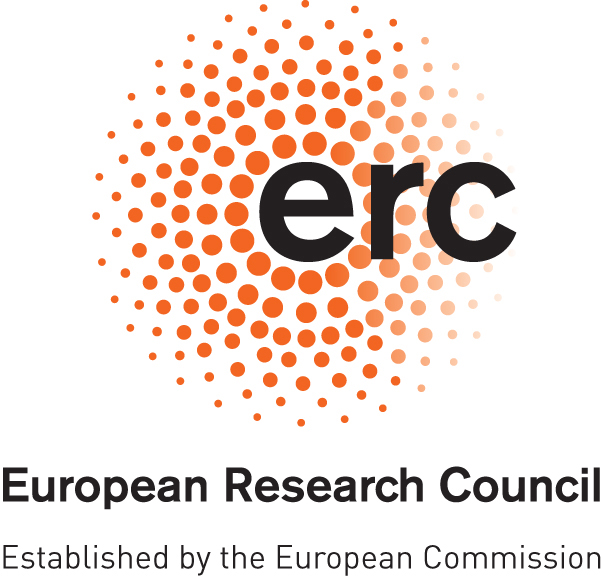}%
\end{textblock}
\begin{textblock}{20}(-2.4, 6.3)
    \includegraphics[width=80px]{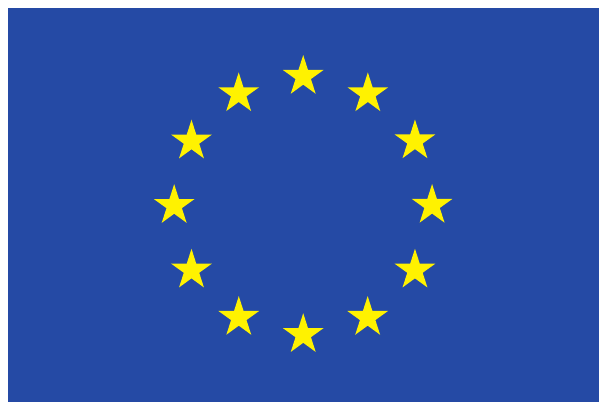}%
\end{textblock}
\newpage
\setcounter{page}{1}
\fi

\section{Introduction}
\label{sec:intro}

The \emph{degeneracy} of a graph $G$, denoted $\deg(G)$, is the smallest integer $d$ such that every subgraph of $G$ contains a vertex of degree at most $d$. A closely related parameter is the \emph{arboricity} $\arb(G)$, defined as the smallest integer $\alpha$ such that the edges of $G$ can be partitioned into $\alpha$ acyclic subgraphs. It is well known that these two parameters are within a multiplicative factor of $2$ from each other, precisely for every graph $G$ with at least one edge, we have
\begin{equation}\label{eq:arbdeg}
    \arb(G)\leq \deg(G)\leq 2\cdot \arb(G)-1.
\end{equation}
For completeness, we include a short proof. For the left inequality, by iteratively removing from the graph a vertex with the smallest degree, we can construct an ordering of vertices such that every vertex has at most $\deg(G)$ larger neighbors. Then coloring the edges with $\deg(G)$ colors so that for every vertex $u$, the edges connecting $u$ with larger neighbors receive pairwise different colors, yields a partition of the edge set into $\deg(G)$ acyclic subgraphs. For the right inequality, if we consider any subgraph $G'$ of $G$, then by restricting the partition witnessing the value of arboricity to the edge set of $G'$ we obtain a partition of $E(G')$ into $\arb(G)$ acyclic subgraphs. Consequently, $G'$ has at most $\arb(G)\cdot (|V(G')|-1)$ edges. Then it cannot happen that all the vertices of $G'$ have degrees at least $2\cdot \arb(G)$, for otherwise the total number of edges in $G'$ would be at least $\frac{2\cdot \arb(G)\cdot |V(G')|}{2}=\arb(G)\cdot |V(G')|$.

Degeneracy and arboricity often serve as parameters measuring uniform sparsity in a graph, because they stipulate that the number of edges is bounded linearly in the number of vertices not only in the whole graph, but also in every subgraph. Although these two notions are essentially equivalent thanks to \eqref{eq:arbdeg}, their popularity differs between communities. In structural graph theory, particularly the field of Sparsity (see e.g.~\cite{nevsetvril2012sparsity}), it is the degeneracy that is considered the basic notion. On the other hand, the area of dynamic and distributed algorithms often focuses on arboricity. We remark that there are two other notions that can be linked to degeneracy and arboricity by inequalities similar to \eqref{eq:arbdeg}: the \emph{maximum average degree} $\mad(G)$, defined as the maximum average degree among the subgraphs of $G$; and the \emph{orientability} $\ort(G)$, defined as the minimum integer $k$ such that the edges of $G$ can be oriented so that every vertex has outdegree at most $k$. See \cite[Chapter~3]{nevsetvril2012sparsity}.

In this work, we study the problem of maintaining an approximate value of the degeneracy of a dynamic graph, updated by edge insertions and deletions, defined as follows.

\begin{tcolorbox}
    \noindent \dyndegany{\beta}{D}
    
    \smallskip
    \noindent\textbf{Initialization:} An edgeless graph $G$ on $n$ vertices, a fixed integer $d_\mx\in \N$.
    
    \smallskip
    \noindent\textbf{Update:} Insert or delete an edge from $G$, subject to the guarantee that $\deg(G)\leq d_\mx$.

    \smallskip
    \noindent\textbf{Query:} After every update, output an integer $d$ satisfying $\deg(G) \leq d \leq \beta\cdot \deg(G) + D$.
\end{tcolorbox}

In \cite{ChristiansenR22}, Christiansen and Rotenberg gave a dynamic data structure for the (analogously defined) \ProblemName{Dynamic Approximate Arboricity} problem. Their data structure reports a number $\alpha$ satisfying $\arb(G)\leq \alpha\leq \arb(G)+2$ after every update, as well as maintains a suitable partition of the edge set of $G$ into $\alpha$ acyclic subgraphs. Every update is handled in worst-case time $\mathrm{poly}(\log n,\alpha_\mx)$, where $\alpha_\mx$ is an upper bound on the arboricity of $G$ guaranteed at all times.

By \eqref{eq:arbdeg}, the result of Christiansen and Rotenberg gives rise to a dynamic data structure for the \dyndegany{2}{3} problem with worst-case update time $\mathrm{poly}(\log n,d_\mx)$. Indeed, if their data structure returns $\alpha$, then it suffices to return $d\coloneqq \max(2\alpha-1,0)$.

The main result of this work is the following conditional lower bound under the Strong Exponential Time Hypothesis (SETH). Intuitively, it asserts that the multiplicative approximation factor of $2$ cannot be improved while maintaining strongly sublinear update time.

\begin{theorem}\label{thm:main}
    Unless SETH fails, for any $\eps, \delta > 0$, any $D, k \in \N$, and any function ${f\colon \N\to \N}$, \mbox{\dyndeg{\eps}{D}} admits no data structure whose initialization time is bounded by ${f(d_\mx)\cdot n^k}$ and whose amortized update time is bounded by $f(d_\mx)\cdot n^{1-\delta}$.
\end{theorem}

Thus, \cref{thm:main} shows that the problems of maintaining an approximation of the arboricity and of the degeneracy have quite different complexities: arboricity can be efficiently maintained up to an additive error of~$2$, while degeneracy requires a multiplicative factor of~$2$.

Our proof of \cref{thm:main} follows the by-now-classic approach to proving fine-grained lower bounds for dynamic problems, pioneered by Abboud and Vassilevska Williams~\cite{AbboudW14}. That is, we show that a data structure for the \dyndeg{\eps}{D} problem achieving a strongly sublinear amortized update time could be used to solve the \ov problem in strongly subquadratic time, thereby refuting the SETH. For the reduction from \ov to \dyndeg{\eps}{D} we introduce an interesting intermediate problem \dyncir: We have a Boolean circuit $C$ whose input bits are being updated over time (but $C$ itself stays fixed), and the task is to efficiently maintain the value of $C$ on the input bits. \ov can be easily reduced to \dyncir by constructing a natural circuit implementing the logic of \ov. To reduce \dyncir to \dyndeg{\eps}{D}, we exploit and strengthen the approach used by Pilipczuk, Siebertz, and Toru\'nczyk~\cite{PilipczukST18,PilipczukST18a} in their proof that deciding whether a graph has degeneracy at most $2$ is a $\mathsf{P}$-hard problem. 

\section{Preliminaries}
\label{sec:prelim}

\paragraph*{Basic notation.} for a positive integer $n$, we denote $[n] \coloneqq \{1, 2, \ldots, n\}$.
If $x \in \Z^d$ is a $d$-dimensional vector, then we write $x(i)$ to refer to its $i$-th coordinate for $i \in [d]$. (This is because $x$ can be equivalently represented as a function $x \colon [d] \to \Z$.)
We use standard asymptotic notation~$\Oh(\cdot)$. Given a set of parameters $P$, we write $\Oh_P(\cdot)$ to denote asymptotic bounds where multiplicative constant factors depend solely on the parameters in $P$.

\paragraph*{Graphs.}
If not stated otherwise, all graphs considered in this paper are undirected and simple.
For a graph $G$, we denote its set of vertices by $V(G)$ and its set of edges by $E(G)$.
We abbreviate an edge $(x, y)$ by $xy$.
Given a vertex $v \in V(G)$, we denote its degree in graph $G$ by $d_G(v)$.
The minimum vertex degree in $G$ is denoted by $\delta(G)$, and the maximum vertex degree by $\Delta(G)$.
For a subset of vertices $U \subseteq V(G)$, we denote the induced subgraph of $G$ on $U$ by $G[U]$. An \emph{independent set} in a graph is a set of pairwise non-adjacent vertices.

For an integer $k \ge 0$, a graph $G$ is \emph{$k$-degenerate} if every non-empty subgraph $H\subseteq G$ contains a vertex of degree at most $k$, i.e., $\delta(H) \le k$.
The \emph{degeneracy} of $G$, denoted $\deg(G)$, is the minimum integer $k$ such that $G$ is $k$-degenerate, which equivalently satisfies $\deg(G) = \max_{H \subseteq G} \delta(H)$.
A vertex ordering $(v_1, v_2, \ldots, v_n)$ of $V(G)$ is a \emph{$k$-degeneracy ordering} of $G$ if $d_{G[\{v_i, \ldots, v_n\}]}(v_i) \le k$ for every $i \in [n]$.
The following standard fact links degeneracy with degeneracy orderings.

\begin{fact}[folklore]\label{fact:degeneracy}
    For any graph $G$ and integer $d \ge 0$, we have $\deg(G) \le d$ if and only if $G$ admits a~$d$-degeneracy ordering.
\end{fact}

For $d = \deg(G)$, any $d$-degeneracy ordering of $G$ is called a \emph{degeneracy ordering} of $G$ or an \emph{elimination ordering} of $G$.

We say that a subset of vertices $U$ of a graph $G$ \emph{$k$-unravels} in $G$ if one can iteratively remove the vertices of $U$ from $G$ so that for every vertex $u\in U$, $u$ has degree at most $k$ in the remaining graph at the time of its removal. Note from \cref{fact:degeneracy} it follows that $G$ is $k$-degenerate if and only if $V(G)$ $k$-unravels in $G$.

\paragraph*{Circuits.}
A \emph{Boolean circuit} (or simply circuit) $C$ is a vertex-labelled acyclic directed graph whose vertices are called \emph{gates}, and edges are called \emph{wires}.
The graph $C$ satisfies the following properties:
\begin{itemize}
\item there is one sink (a vertex of out-degree 0) which is called the \emph{output gate}. The output gate has only one incoming wire called the \emph{output wire};
\item there are possibly multiple sources (vertices of in-degree 0) which are called \emph{input gates}, denoted $\cinputs(C)$; and
\item every gate except the input gates and the output gate is labelled with one of the Boolean functions: AND, OR,~NOT. Both AND-gates and OR-gates may have arbitrarily many incoming wires, while every NOT-gate is required to have only one incoming wire.
\end{itemize}

We write $\cgates(C) \coloneqq V(C)$ and $\cwires(C) \coloneqq E(C)$.
The size of a circuit $C$, denoted $|C|$, is defined as $|C| \coloneqq |\cgates(C)| + |\cwires(C)|$.

An input vector to a circuit $C$ is a Boolean assignment $x \colon \cinputs(C) \to \{0, 1\}$.
Given an input vector $x$, one can naturally evaluate all the gates of the circuit by traversing them in topological order; we denote the resulting Boolean value at gate $g$ under input $x$ by $g(x)$.
The value returned by the output gate is said to be the output of $C$ on $x$, abbreviated as $C(x)$.
Two circuits $C$ and $C'$ are said to be \emph{equivalent} if there is a bijection $f \colon \cinputs(C') \to \cinputs(C)$ such that for every input vector $x$ to $C$ we have $C(x) = C'(x \circ f)$.

We say that a circuit is \emph{monotone} if it contains only AND and OR gates.
Monotone circuits are monotonic in the logical sense: flipping one of the input bits from $0$ to $1$ can only change the output from $0$ to $1$, never from $1$ to $0$.
A circuit $C$ is said to be a $(2,2)$-circuit if every gate in $C$ has in-degree at most $2$ (at most $2$ input wires) and out-degree at most $2$ (at most $2$ output wires).
Note that both AND and OR gates with only one incoming wire act as a copy gate.

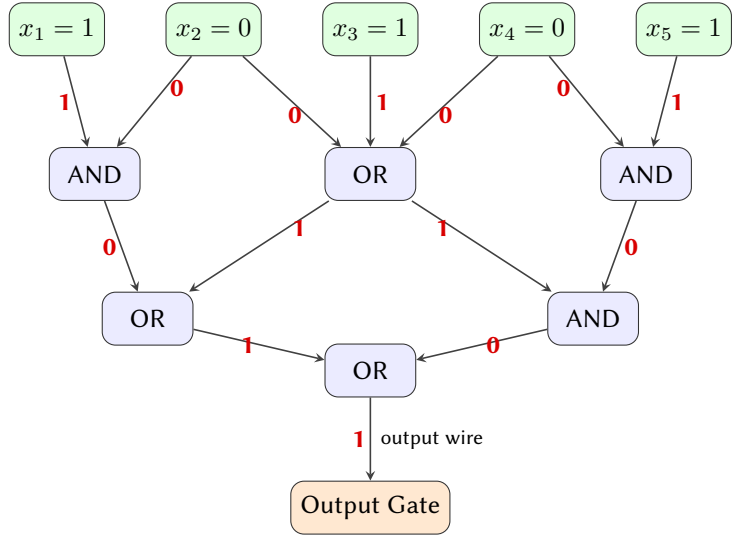
\begin{figure}[ht]
\centering
\begin{tikzpicture}[
    node distance=1.2cm and 0.8cm,
    gate/.style={rectangle, rounded corners=5pt, draw=black!80, fill=blue!8, minimum width=1.2cm, minimum height=0.7cm, font=\small\sffamily, align=center},
    input/.style={rectangle, rounded corners=5pt, draw=black!80, fill=green!12, minimum width=1.0cm, minimum height=0.7cm, font=\small\sffamily, align=center},
    output/.style={rectangle, rounded corners=5pt, draw=black!80, fill=orange!18, minimum width=1.5cm, minimum height=0.7cm, font=\small\sffamily, align=center},
    wire/.style={->, >=stealth, semithick, draw=black!75},
    bit/.style={font=\small\sffamily\bfseries, color=red!85!black, inner sep=1.5pt}
]

    \node[input] (p1) {$x_1 = 1$};
    \node[input, right=of p1] (p2) {$x_2 = 0$};
    \node[input, right=of p2] (p3) {$x_3 = 1$};
    \node[input, right=of p3] (p4) {$x_4 = 0$};
    \node[input, right=of p4] (p5) {$x_5 = 1$};

    \node[gate, below=of p1, xshift=0.5cm] (g1) {AND};
    \node[gate, below=of p3] (g2) {OR}; 
    \node[gate, below=of p5, xshift=-0.5cm] (g3) {AND};

    \node[gate, below=of g1, xshift=0.7cm] (g4) {OR};
    \node[gate, below=of g3, xshift=-0.7cm] (g5) {AND};

    \node[gate, below=1.0cm of p3, yshift=-2.8cm] (g6) {OR};

    \node[output, below=1.1cm of g6] (out) {Output Gate};

    \draw[wire] (p1) -- node[bit, left] {1} (g1);
    \draw[wire] (p2) -- node[bit, right, pos=0.35] {0} (g1);

    \draw[wire] (p2) -- node[bit, left, pos=0.65] {0} (g2);
    \draw[wire] (p3) -- node[bit, right] {1} (g2);
    \draw[wire] (p4) -- node[bit, right, pos=0.65] {0} (g2);

    \draw[wire] (p4) -- node[bit, left, pos=0.35] {0} (g3);
    \draw[wire] (p5) -- node[bit, right] {1} (g3);

    \draw[wire] (g1) -- node[bit, left] {0} (g4);
    \draw[wire] (g2) -- node[bit, right, pos=0.3] {1} (g4);
    \draw[wire] (g2) -- node[bit, left, pos=0.3] {1} (g5);
    \draw[wire] (g3) -- node[bit, right] {0} (g5);

    \draw[wire] (g4) -- node[bit, left] {1} (g6);
    \draw[wire] (g5) -- node[bit, right] {0} (g6);

    \draw[wire] (g6) -- node[bit, left] {1} node[right, font=\scriptsize\sffamily, color=black] {output wire} (out);

\end{tikzpicture}
\caption{An example Boolean circuit $C$. The circuit is monotone (contains no NOT gates) but is not a $(2,2)$-circuit because it contains a gate with in-degree $3$ (the OR-gate in the first layer). Red labels indicate the bit values propagated along wires under the given input vector.}
\label{fig:example-circuit}
\end{figure}

\begin{fact}
    \label{fact:circuit22}
    For every circuit $C$ there exists an equivalent $(2,2)$-circuit $C'$ of size $\Oh(|C|)$.
    Additionally, if $C$ is monotone then so is $C'$.
\end{fact}

\begin{proof}[Sketch of proof]
If any gate $g \in \cgates(C)$ has $k > 2$ output wires going to gates $g_1, \ldots, g_k$, we replace these wires with a binary tree $T$ such that: (1) $T$ is rooted at $g$ and all its leaves are $g_1, \ldots, g_k$; (2) every internal node of $T$ is an AND gate with one input wire and two outgoing wires, thus the size of $T$ is $\Theta(k)$.
Similar transformations can be performed for AND / OR gates which have more than two incoming wires.
\end{proof}

\paragraph*{SETH and OV.}
The \emph{Strong Exponential Time Hypothesis} (SETH), introduced by Impagliazzo and Paturi~\cite{impagliazzo2001complexity}, postulates that $k$-SAT (the satisfiability problem for Boolean formulas in $k$-CNF) cannot be solved significantly faster than exhaustive search as $k$ grows. Formally, for $k \ge 3$, let $s_k \in [0, 1]$ be the infimum of all real numbers $\delta$ such that $k$-SAT on $n$ variables can be solved in time $\Oh(2^{\delta n})$. SETH asserts that:
\[
\lim_{k \to \infty} s_k = 1.
\]
Equivalently, SETH asserts that for every $\eps > 0$, there exists an integer $k \ge 3$ such that $k$-SAT cannot be solved in time $\Oh((2-\eps)^n)$.

While SETH was originally formulated to rule out faster algorithms for exact exponential-time problems, fine-grained complexity leverages it to prove lower bounds for polynomial-time dynamic and static problems. The standard stepping stone for translating SETH into lower bounds for problems in $\mathsf{P}$ is the \ProblemName{Orthogonal Vectors} (OV) problem.

\begin{tcolorbox}
    \noindent \ov
    
    \smallskip
    \noindent\textbf{Given:} an integer $d$ and two sets $A, B \subseteq \{0,1\}^d$, each consisting of $n$ vectors.

    \smallskip
    \noindent\textbf{Question:} do there exist vectors $a \in A$ and $b \in B$ such that $\scalar{a}{b} = 0$?
\end{tcolorbox}

\begin{fact}[\cite{williams2005new}]
    \label{fact:ov-lb}
    Unless SETH fails, for any $\eps, \delta > 0$, \ov in dimension $d \leq \Oh(n^{\delta})$ cannot be solved in time $\Oh(n^{2-\eps})$.
\end{fact}

\section{From Orthogonal Vectors to Dynamic Circuit Evaluation}
\label{sec:ov-to-circuit}

Let us define the problem of \dyncir which appears as an intermediate problem in our proof of \cref{thm:main}.

\begin{tcolorbox}
    \noindent \dyncir
    
    \smallskip
    \noindent\textbf{Initialization:} a fixed Boolean circuit $C$ and an input vector of bits $x$.
    
    \smallskip
    \noindent\textbf{Update:} flip one of the bits in the input vector $x$.

    \smallskip
    \noindent\textbf{Query:} after each update, report the value of $C(x)$.
\end{tcolorbox}

In this section, we show that SETH-based hardness of \ov (\cref{fact:ov-lb}) rules out the existence of an efficient data structure for \dyncir.
The following lemma establishes a desired lower bound.

\begin{lemma}
    \label{lem:circuit-lower-bound}
    Unless SETH fails, for any $k \in \N$ and $\eps > 0$, \dyncir admits no data structure with initialization time $\Oh(m^k)$ and amortized update time $\Oh(m^{1-\eps})$, where $m$ is the size of the circuit.
    The result holds even if the problem is restricted to monotone $(2,2)$-circuits.
\end{lemma}

We begin with an auxiliary result which is the crux of the reduction.

\begin{lemma}
    \label{lem:ov-to-circuit}
    There is a data structure $\Dovtocirc$ that supports the following methods:
    \begin{itemize}
        \item $\oper{init}(n, d)$: Given two integers $n, d \in \N$, fixes the values of $n$ and $d$ for the entire run, and returns a monotone Boolean $(2,2)$-circuit $C_{n, d}$ with $\Oh(n d)$ gates.
        Runs in time $\Oh(n d)$.
        \item $\oper{ovToCirc}(A, B)$: Given two sets $A, B \subseteq \{0,1\}^d$, each of size $n$, returns a sequence\footnote{Formally, $\oper{ovToCirc}(A, B)$ returns a vector $x_1$ and a sequence of sets $F_2, \ldots, F_n$ such that the vectors $x_2, \ldots, x_n$ can be obtained from $x_1$ by flipping the bits on positions $F_2, \ldots, F_n$, respectively. We assume it returns all vectors $x_1, \ldots, x_n$ for simplicity of further notation.} $(x_1, \ldots, x_n)$ of input vectors to the circuit $C_{n, d}$ satisfying the following properties:
        \begin{itemize}
            \item for every $k \in [n]$, vectors $x_k$ and $x_1$ differ on at most $2$ positions; and
            \item $C_{n, d}(x_k) = 0$ for some $k \in [n]$ if and only if there exist vectors $a \in A$ and $b \in B$ such that $\scalar{a}{b} = 0$, that is, if $(d, A, B)$ is a~\yesinstance of \ov.
        \end{itemize}
        $\oper{OvToCirc}$ runs in time $\Oh(n d)$ as well.
    \end{itemize}
\end{lemma}

\begin{proof}
{
    \newcommand{\inputsimple}{\oper{input}}
    \newcommand{\inputswitch}{\oper{switch}}
    \newcommand{\inputprim}{\oper{input'}}
    \newcommand{\choice}{\oper{choice}}
    \newcommand{\compare}{\oper{compare}}
    \newcommand{\orthogonal}{\oper{orthogonal}}
    \newcommand{\result}{\oper{output}}

    Let us describe the circuit $C_{n, d}$ returned by $\oper{init}(n, d)$.
    First, we provide an auxiliary monotone circuit $C_{n,d}'$ that is not necessarily a $(2,2)$-circuit.
    The set of input gates consists of three families of sets:
    \begin{align*}
        \cinputs(C_{n,d}') \coloneqq & \phantom{\cup}\,\ \{\inputsimple^A_{i,j} \mid i \in [n], j \in [d]\} \\
        & \cup \{\inputsimple^B_{i, j} \mid i \in [n], j \in [d]\} \\
        & \cup \{\inputswitch^A_{i} \mid i \in [n] \}.
    \end{align*}

    The remaining gates of $C_{n,d}'$ are defined as follows (see \cref{fig:example-ov-to-circ} for an example circuit).
    \begin{itemize}
        \item $\inputprim^A_{i,j} \coloneqq \text{AND}(\inputsimple^A_{i,j}, \inputswitch^A_i)$ for every $i \in [n]$,
        \item $\choice^A_{j} \coloneqq \text{OR}(\inputprim^A_{1, j}, \inputprim^A_{2, j}, \ldots, \inputprim^A_{n, j})$ for every $j \in [d]$,
        \item $\compare_{i,j} \coloneqq \text{AND}(\choice^A_j, \inputsimple^B_{i, j})$ for every $i \in [n]$ and $j \in [d]$,
        \item $\orthogonal_{i} \coloneqq \text{OR}(\compare_{i, 1}, \compare_{i, 2}, \ldots, \compare_{i, d})$ for every $i \in [n]$,
        \item $\result \coloneqq \text{AND}(\orthogonal_1, \orthogonal_2, \ldots, \orthogonal_n)$.
    \end{itemize}

    \begin{figure}[ht]
\centering
\begin{tikzpicture}[
    scale=0.75, transform shape,
    in_node/.style={draw=blue!60!black, fill=blue!8, rectangle, rounded corners=3pt, inner sep=3pt, font=\footnotesize, align=center},
    sw_node/.style={draw=orange!80!black, fill=orange!12, rectangle, rounded corners=3pt, inner sep=3pt, font=\footnotesize, align=center},
    gate_and/.style={draw=purple!80!black, fill=purple!10, rectangle, rounded corners=3pt, inner sep=3pt, font=\footnotesize, align=center},
    gate_or/.style={draw=teal!80!black, fill=teal!10, rectangle, rounded corners=3pt, inner sep=3pt, font=\footnotesize, align=center},
    gate_out/.style={draw=red!80!black, fill=red!15, rectangle, rounded corners=3pt, inner sep=4pt, font=\small\bfseries, align=center},
    box_group/.style={draw=gray!40, dashed, rounded corners=6pt, fill=gray!8, inner sep=8pt},
    box_vector/.style={draw=blue!50!black, dotted, rounded corners=4pt, fill=blue!4, inner sep=3pt},
    wire/.style={gray!40, ->, >=stealth, thin}
]

    \node[in_node] (inA11) at (-12.0, 0) {$\operatorname{input}^A_{1,1}$};
    \node[in_node] (inA12) at (-10.6, 0) {$\operatorname{input}^A_{1,2}$};
    \node[in_node] (inA13) at (-9.2, 0)  {$\operatorname{input}^A_{1,3}$};
    \node[in_node] (inA21) at (-7.4, 0)  {$\operatorname{input}^A_{2,1}$};
    \node[in_node] (inA22) at (-6.0, 0)  {$\operatorname{input}^A_{2,2}$};
    \node[in_node] (inA23) at (-4.6, 0)  {$\operatorname{input}^A_{2,3}$};

    \node[sw_node] (swA1) at (-2.7, 0) {$\operatorname{switch}^A_1$};
    \node[sw_node] (swA2) at (-1.3, 0) {$\operatorname{switch}^A_2$};

    \node[in_node] (inB11) at (0.5, 0) {$\operatorname{input}^B_{1,1}$};
    \node[in_node] (inB12) at (2.2, 0) {$\operatorname{input}^B_{1,2}$};
    \node[in_node] (inB13) at (3.9, 0) {$\operatorname{input}^B_{1,3}$};
    \node[in_node] (inB21) at (6.1, 0) {$\operatorname{input}^B_{2,1}$};
    \node[in_node] (inB22) at (7.8, 0) {$\operatorname{input}^B_{2,2}$};
    \node[in_node] (inB23) at (9.5, 0) {$\operatorname{input}^B_{2,3}$};

    \node[box_group, fit=(inA11)(inA23), label=above:{\small Matrix $A$ Inputs}] (boxA) {};
    \node[box_group, fit=(swA1)(swA2), label=above:{\small Choice Switches}] (boxSW) {};
    \node[box_group, fit=(inB11)(inB23), label=above:{\small Matrix $B$ Inputs}] (boxB) {};

    \node[box_vector, fit=(inA11)(inA13), label={[anchor=north, yshift=-5pt]north:{\small $a_1$}}] (box_a1) {};
    \node[box_vector, fit=(inA21)(inA23), label={[anchor=north, yshift=-5pt]north:{\small $a_2$}}] (box_a2) {};
    \node[box_vector, fit=(inB11)(inB13), label={[anchor=north, yshift=-5pt]north:{\small $b_1$}}] (box_b1) {};
    \node[box_vector, fit=(inB21)(inB23), label={[anchor=north, yshift=-5pt]north:{\small $b_2$}}] (box_b2) {};

    \node[gate_and] (inP11) at (-12.0, -2.4) {$\operatorname{input'}^A_{1,1}$\\\textcolor{purple!80!black}{\tiny [AND]}};
    \node[gate_and] (inP12) at (-10.6, -2.4) {$\operatorname{input'}^A_{1,2}$\\\textcolor{purple!80!black}{\tiny [AND]}};
    \node[gate_and] (inP13) at (-9.2, -2.4)  {$\operatorname{input'}^A_{1,3}$\\\textcolor{purple!80!black}{\tiny [AND]}};
    \node[gate_and] (inP21) at (-7.4, -2.4)  {$\operatorname{input'}^A_{2,1}$\\\textcolor{purple!80!black}{\tiny [AND]}};
    \node[gate_and] (inP22) at (-6.0, -2.4)  {$\operatorname{input'}^A_{2,2}$\\\textcolor{purple!80!black}{\tiny [AND]}};
    \node[gate_and] (inP23) at (-4.6, -2.4)  {$\operatorname{input'}^A_{2,3}$\\\textcolor{purple!80!black}{\tiny [AND]}};

    \node[gate_or] (ch1) at (-9.7, -4.4) {$\operatorname{choice}^A_1$\\\textcolor{teal!80!black}{\tiny [OR]}};
    \node[gate_or] (ch2) at (-8.3, -4.4) {$\operatorname{choice}^A_2$\\\textcolor{teal!80!black}{\tiny [OR]}};
    \node[gate_or] (ch3) at (-6.9, -4.4) {$\operatorname{choice}^A_3$\\\textcolor{teal!80!black}{\tiny [OR]}};

    \node[gate_and] (cmp11) at (0.5, -7.6) {$\operatorname{compare}_{1,1}$\\\textcolor{purple!80!black}{\tiny [AND]}};
    \node[gate_and] (cmp12) at (2.2, -7.6) {$\operatorname{compare}_{1,2}$\\\textcolor{purple!80!black}{\tiny [AND]}};
    \node[gate_and] (cmp13) at (3.9, -7.6) {$\operatorname{compare}_{1,3}$\\\textcolor{purple!80!black}{\tiny [AND]}};

    \node[gate_and] (cmp21) at (6.1, -7.6) {$\operatorname{compare}_{2,1}$\\\textcolor{purple!80!black}{\tiny [AND]}};
    \node[gate_and] (cmp22) at (7.8, -7.6) {$\operatorname{compare}_{2,2}$\\\textcolor{purple!80!black}{\tiny [AND]}};
    \node[gate_and] (cmp23) at (9.5, -7.6) {$\operatorname{compare}_{2,3}$\\\textcolor{purple!80!black}{\tiny [AND]}};

    \node[gate_or] (orth1) at (2.2, -9.6) {$\operatorname{orthogonal}_1$\\\textcolor{teal!80!black}{\tiny [OR]}};
    \node[gate_or] (orth2) at (7.8, -9.6) {$\operatorname{orthogonal}_2$\\\textcolor{teal!80!black}{\tiny [OR]}};

    \node[gate_out] (out) at (5.0, -11.4) {$\operatorname{output}$\\\textcolor{red!80!black}{\tiny [AND]}};

    \draw[wire] (inA11) -- (inP11); \draw[wire] (inA12) -- (inP12); \draw[wire] (inA13) -- (inP13);
    \draw[wire] (inA21) -- (inP21); \draw[wire] (inA22) -- (inP22); \draw[wire] (inA23) -- (inP23);

    \draw[wire] (swA1) -- (inP11); \draw[wire] (swA1) -- (inP12); \draw[wire] (swA1) -- (inP13);
    \draw[wire] (swA2) -- (inP21); \draw[wire] (swA2) -- (inP22); \draw[wire] (swA2) -- (inP23);

    \draw[wire] (inP11) -- (ch1); \draw[wire] (inP21) -- (ch1);
    \draw[wire] (inP12) -- (ch2); \draw[wire] (inP22) -- (ch2);
    \draw[wire] (inP13) -- (ch3); \draw[wire] (inP23) -- (ch3);

    \draw[wire] (ch1) -- (cmp11); \draw[wire] (ch1) -- (cmp21);
    \draw[wire] (ch2) -- (cmp12); \draw[wire] (ch2) -- (cmp22);
    \draw[wire] (ch3) -- (cmp13); \draw[wire] (ch3) -- (cmp23);

    \draw[wire] (inB11) -- (cmp11); \draw[wire] (inB12) -- (cmp12); \draw[wire] (inB13) -- (cmp13);
    \draw[wire] (inB21) -- (cmp21); \draw[wire] (inB22) -- (cmp22); \draw[wire] (inB23) -- (cmp23);

    \draw[wire] (cmp11) -- (orth1); \draw[wire] (cmp12) -- (orth1); \draw[wire] (cmp13) -- (orth1);
    \draw[wire] (cmp21) -- (orth2); \draw[wire] (cmp22) -- (orth2); \draw[wire] (cmp23) -- (orth2);

    \draw[wire] (orth1) -- (out);
    \draw[wire] (orth2) -- (out);

\end{tikzpicture}
\caption{Structure of the circuit $C'_{n,d}$ constructed in \cref{lem:ov-to-circuit} for $n=2$ and $d=3$.}
\label{fig:example-ov-to-circ}
\end{figure}
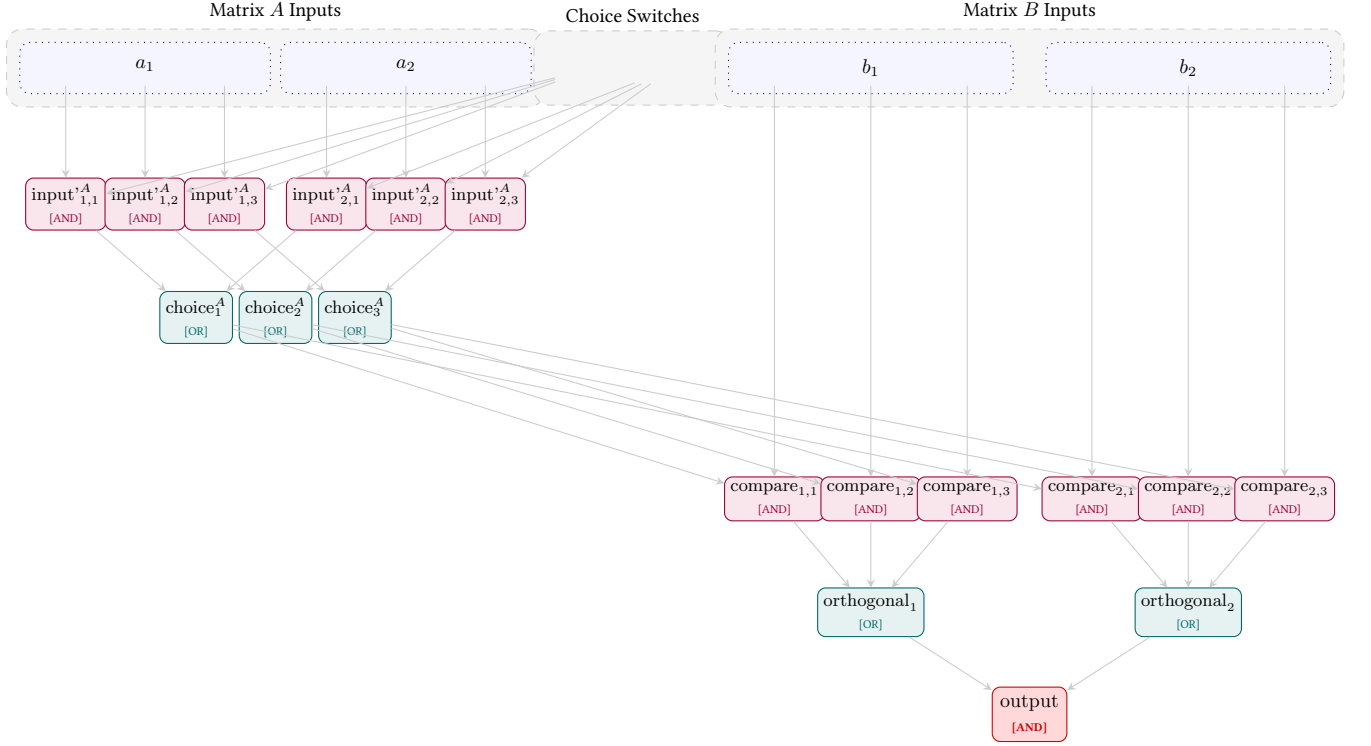

    The size of the circuit $C_{n,d}'$ is $\Oh(nd)$.
    Finally, by applying \cref{fact:circuit22} we can obtain the desired monotone $(2,2)$-circuit $C_{n,d}$ of size $\Oh(nd)$ as well.

    Now, let us describe vectors $(x_1, \ldots, x_n)$ returned by $\oper{ovToCirc}(A, B)$.
    Let $A = \{a_1, \ldots, a_n\}$ and $B = \{b_1, \ldots, b_n\}$.
    For every $k \in [n], i \in [n], j \in [d]$ we set
    \begin{align*}
        x_k(\inputsimple^A_{i, j}) & \coloneqq a_i(j); \\
        x_k(\inputsimple^B_{i, j}) & \coloneqq b_i(j); \\
        x_k(\inputswitch^A_i) & \coloneqq \begin{cases}
         1 & \text{if $i = k$,} \\
         0 & \text{if $i \neq k$.}
        \end{cases}
    \end{align*}
    First, observe that for $k, k' \in [n]$ vectors $x_k$ and $x_{k'}$ can differ only on input gates $\inputswitch^A_k$ and $\inputswitch^A_{k'}$.
    Hence, instead of returning all the vectors $x_1, \ldots, x_n$ (whose total bit size is $\Oh(n^2d)$) it is enough to return the vector $x_1$ (of bit size $\Oh(nd)$) and $n-1$ sets of size at most $2$ representing the differences between $x_i$ and $x_1$, for $i = 2, \ldots, n$. 

    It remains to be proved that circuit $C_{n,d}$ together with vectors $x_1, \ldots, x_n$ have the desired property related to orthogonal vectors in sets $A$ and $B$.
    To this end, we prove the following.

    \begin{claim}\label{claim:ov-to-circ-equiv}
    For every $k \in [n]$, $C_{n,d}(x_k) = 0$ if and only if there exists $b \in B$ such that $\scalar{a_k}{b} = 0$.
    \end{claim}

    \begin{claimproof}
    Fix an integer $k \in [n]$.
    The purpose of input gates $\inputswitch^A_{(\cdot)}$ is to pick one of the vectors from set $A$.
    Indeed, observe that by definition, for every $i \in [n]$ and $j \in [d]$,
    \begin{align*}
        \inputprim^A_{i,j}(x_k) & = \inputsimple^A_{i,j}(x_k) \land \inputswitch^A_i(x_k) \\
        & = a_i(j) \land \inputswitch^A_i(x_k) = \begin{cases}
         a_k(j) & \text{if $i = k$,} \\
         0 & \text{if $i \neq k$.}
        \end{cases}
    \end{align*}
    Consequently, for every $j \in [d]$ we have
    \[
    \choice_j^A(x_k) = \inputprim^A_{1, j}(x_k) \lor \inputprim^A_{2, j}(x_k) \lor \ldots \lor \inputprim^A_{n, j}(x_k) = a_k(j),
    \]
    so the vector $a_k$ is \emph{copied} to the gates $\choice^A_{(\cdot)}$.

    Next, for every $i \in [n]$ and $j \in [d]$,
    \[
        \compare_{i,j}(x_k) = \choice^A_j(x_k) \land \inputsimple^B_{i, j}(x_k) = a_k(j) \land b_i(j)
    \]
    Hence, if we take an OR over all values $\compare_{i,(\cdot)}(x_k)$ we obtain $1$ if and only if vectors $a_k$ and $b_i$ are \emph{not} orthogonal to each other.
    This is done precisely by gates $\orthogonal_i$ for every $i \in [n]$.

    Finally, the output gate is an AND over all the gates $\orthogonal_i$, so $\result(x_k) = 0$ if and only if $\orthogonal_i(x_k) = 0$ for some $i \in [n]$. This is equivalent to saying that vectors $a_k$ and $b_i$ are orthogonal for some $i \in [n]$.
    This ends the proof of the claim.
    \end{claimproof}

    \cref{claim:ov-to-circ-equiv} immediately implies that $C_{n,d}(x_k) = 0$ for some $k \in [n]$ if and only if $(d, A, B)$ is a~\yesinstance of \ov.
    This concludes the proof.
}
\end{proof}

Having proved \cref{lem:ov-to-circuit} we are ready to show the promised SETH-based lower bound for \dyncir.

\begin{proof}[Proof of \cref{lem:circuit-lower-bound}]
    Suppose that, for some $k \in \N$ and $\eps \in (0, 1]$, there exists a data structure $\Dcirc$ for \dyncir with initialization time $\Oh(m^k)$ and amortized update time $\Oh(m^{1-\eps})$.
    We will show that using $\Dcirc$ one can solve $\ov$ in dimension $d \leq \Oh(n^{\delta})$ in time $\Oh(n^{2-\eps'})$, where $\delta \coloneqq \frac{\eps}{2k}$ and $\eps' \coloneqq \frac{\eps^2}{2k}$.
    This will contradict \cref{fact:ov-lb}.

    Let $(d, A, B)$ be an instance of the $\ov$ problem.
    Denote $n \coloneqq |A| = |B|$, and assume that $d \leq \Oh(n^{\delta})$, where $\delta = \frac{\eps}{2k}$.
    Our algorithm proceeds as follows.

    Let $\Dovtocirc$ be the data structure from \cref{lem:ov-to-circuit}.
    First, we initialize $\Dovtocirc$ with parameters $(n', d)$, where $n' \coloneqq \ceil{n^{\frac{1}{k}}}$.
    The call $\oper{init}(n', d)$ returns a monotone circuit $C_{n', d}$ with $m \in \Oh(n' d)$ gates.
    Second, we initialize $\Dcirc$ with the circuit $C_{n', d}$.
    The total running time of both initialization steps is
    \[
        \Oh(n' d) + \Oh(m^k) \leq \Oh_k\left((n'd)^k\right) \leq \Oh_k\left((n^{\frac{1}{k} + \frac{\eps}{2k}})^k\right) = \Oh_k\left(n^{1 + \frac{\eps}{2}}\right) \leq \Oh \left(n^{2-\eps'}\right),
    \]
    since $k$ is a universal constant here.

    Next, we cover the sets $A$ and $B$ with their subsets: $A = A_1 \cup \ldots \cup A_s$ and $B = B_1 \cup \ldots \cup B_s$, where $s \leq \Oh(n^{1-\frac{1}{k}})$ and for every $i \in [s]$, $|A_i| = |B_i| = n' = \ceil{n^{\frac{1}{k}}}$.
    Note that since all the sets $A_i$ are required to be of the same size, they are not necessarily disjoint.
    We iterate over all possible values of  $i, j \in [s]$.
    Having $(i, j)$ fixed, we want to find out whether the sets $A_i$ and $B_j$ contain two vectors orthogonal to each other.
    To this end, we run $\oper{ovToCirc}(A_i, B_j)$ on $\Dovtocirc$.
    This call runs in time $\Oh(n'd)$, and returns a sequence $(x_1, \ldots, x_{n'})$ of input vectors to $C_{n',d}$.
    By~\cref{lem:ov-to-circuit}, to establish whether $(d, A_i, B_j)$ is a~\yesinstance of \ov it is enough to evaluate $C_{n', d}(x_k)$ for every $k \in [n']$.
    Now, we make use of the data structure $\Dcirc$. First, we evaluate $C_{n',d}(x_1)$ by setting all input bits to match $x_1$, which requires at most $|\cinputs(C_{n',d})| \in \Oh(n'd)$ bit flips from the previous state.
    Then, we change the input to $x_2, \ldots, x_{n'}$ consecutively, and by \cref{lem:ov-to-circuit} each such update requires flipping at most $2+2=4$ bits.
    Hence, to compute all the outputs $C_{n', d}(x_1), \ldots, C_{n', d}(x_{n'})$ it is enough to make $\Oh(n'd) + \Oh(n') \leq \Oh(n'd)$ updates on $\Dcirc$.
    The running time of these updates is upper bounded by
    \[
        \Oh(n'd) \cdot \Oh(m^{1-\eps}) \leq \Oh\left((n'd)^{2-\eps}\right) \leq \Oh\left((n^{\frac{1}{k} + \frac{\eps}{2k}})^{2-\eps}\right) = \Oh\left(n^{\frac{2}{k} - \frac{\eps}{k} + \frac{\eps}{k} - \frac{\eps^2}{2k}}\right) = \Oh\left(n^{\frac{2}{k} - \frac{\eps^2}{2k}}\right).
    \]
    Since we conduct the checks above for each choice of $i, j \in [s]$, the total running time of this phase is at most
    \[
        s^2 \cdot \Oh\left(n^{\frac{2}{k} - \frac{\eps^2}{2k}}\right) \leq \Oh\left((n^{1-\frac{1}{k}})^2\right) \cdot \Oh\left(n^{\frac{2}{k} - \frac{\eps^2}{2k}}\right) = \Oh\left(n^{2-\frac{2}{k} + \frac{2}{k} - \frac{\eps^2}{2k}}\right) = \Oh \left(n^{2-\eps'}\right),
    \]
    as desired.
    Finally, if $(d, A_i, B_j)$ is a~\noinstance for every $i, j \in [s]$, we can report that $(d, A, B)$ is a~\noinstance as well.
    The correctness of our algorithm is straightforward, and its running time contradicts \cref{fact:ov-lb}. This completes the proof.
\end{proof}

\section{From Dynamic Circuit Evaluation to Dynamic Degeneracy}
\label{sec:circuit-to-deg}

In this section, we prove the main result of our work, namely \cref{thm:main}.
To this end, we provide an efficient reduction from \dyncir, whose conditional hardness was shown in \cref{sec:ov-to-circuit}, to \dyndeg{\eps}{D}.
We start with one more auxiliary fact that will simplify descriptions of gadgets later on.

\begin{lemma}
\label{lem:circ-to-deg-helper}
Let $C$ be a monotone Boolean $(2,2)$-circuit of size $m$ and $x \colon \cinputs(C) \to \{0,1\}$ be an input vector.
Then, there exist a monotone Boolean $(2,2)$-circuit $C'$ of size $\Oh(m)$ and an input vector $x' \colon \cinputs(C') \to \{0,1\}$ such that:
\begin{itemize}
    \item $C(x) = C'(x')$;
    \item every input gate of $C'$ has outdegree $1$; and
    \item every OR-gate of $C'$ has indegree $2$ and outdegree $1$.
\end{itemize}
Both $C'$ and $x'$ can be computed in time $\Oh(m)$. Moreover, if $x$ is a dynamic vector subject to bit flips, then every bit flip in $x$ can be translated to $\Oh(1)$ bit flips in $x'$ such that the equality $C(x) = C'(x')$ is preserved.
\end{lemma}

\begin{proof}
We construct $C'$ and $x'$ from $C$ and $x$ by applying local modifications to the input gates and the OR-gates of $C$.

For every input gate $p \in \cinputs(C)$:
\begin{itemize}
    \item if $\operatorname{outdeg}(p) = 0$, we remove $p$ from $C$ entirely, as an isolated input does not affect the logical evaluation of $C$;
    \item if $\operatorname{outdeg}(p) = 1$, we retain $p$ in $C'$ with $x'(p) = x(p)$;
    \item if $\operatorname{outdeg}(p) = 2$ with outgoing wires $pu_1$ and $pu_2$, we replace $p$ with two independent input gates $p_1$ and $p_2$, setting $x'(p_1) = x'(p_2) = x(p)$ and replacing the wires with $p_1u_1$ and $p_2u_2$. 
\end{itemize}
Consequently, every input gate in $C'$ has outdegree exactly $1$.

For every OR-gate $h \in \cgates(C)$:
\begin{itemize}
    \item if $h$ has indegree $1$, we replace $h$ with an AND-gate $h'$ having the same ingoing and outgoing wires. Note that OR and AND gate on a single input are equivalent: they just copy the incoming bit;
    \item if $h$ has outdegree $2$, we replace $h$ with an OR-gate $h_1$ (with outdegree $1$) feeding into an AND-gate $h_2$ (with outdegree $2$). Gate $h_1$ inherits the incoming wires of $h$, sends its output to $h_2$, while $h_2$ inherits the outgoing wires of $h$.
\end{itemize}

It is easy to verify that these substitutions preserve the circuit evaluation at every step, ensuring $C(x) = 1$ if and only if $C'(x') = 1$.
Each gate in $C$ is replaced by at most $\Oh(1)$ gates and wires, so $|C'| \leq \Oh(m)$, and both $C'$ and $x'$ can be computed in $\Oh(m)$ time.
\end{proof}

\cref{lem:circ-to-deg} encapsulates the crucial ingredients of the reduction from \dyncir to \dyndeg{\eps}{D}.

\begin{lemma}
\label{lem:circ-to-deg}
For every integer $\ell \geq 2$, there exists a data structure $\Dcirctodeg_\ell$ that given a~fixed upon initialization monotone Boolean $(2,2)$-circuit $C$ and a~dynamic input vector $x$, maintains a graph $G_{C, x}$ that satisfies the conditions:
\begin{itemize}
\item $C(x) = 1$ if and only if $\deg(G_{C, x}) \leq \ell$; and
\item $C(x) = 0$ if and only if $2\ell -1 \leq \deg(G_{C, x}) \leq \maxdeg$.
\end{itemize}
The data structure $\Dcirctodeg_\ell$ supports the following methods:
\begin{itemize}
\item $\oper{init}(C, x_0)$: Given a monotone Boolean $(2,2)$-circuit $C$ of size $m$ and an initial input vector $x_0$ returns an initial graph $G_{C, x_0}$, where $|V(G_{C, x_0})| \leq \Oh(\ell^2 \cdot m)$ and $|E(G_{C, x_0})| \leq \Oh(\ell^3 \cdot m)$.
Runs in time $\Oh(\ell^3 \cdot m)$.
\item $\oper{flip}(p)$: Let $C$ and $x$ be the circuit for which the data structure was initialized and the current input vector.
This method returns a graph $G_{C, x'}$ as a sequence of $\Oh(\ell^2)$ edge updates to $G_{C, x}$ where $x'$ is the vector obtained from $x$ by flipping its bit corresponding to $p$, i.e. $x'(p) = 1 - x(p)$, and $x'(q) = x(q)$ for $q \neq p$.
Additionally, $x'$ is set as the new input vector to $C$.
Method $\oper{flip}(p)$ runs in time $\Oh(\ell^2)$.
\end{itemize}
\end{lemma}

\begin{proof}
{

\newcommand{\Vwire}{V^{\text{wire}}}
\newcommand{\Vin}{V^{\text{in}}}
\newcommand{\Vout}{V^{\text{out}}}
\newcommand{\Vor}{V^{\text{or}}}
\newcommand{\Vand}{V^{\text{and}}}
\newcommand{\Ein}{E^{\text{in}}}
\newcommand{\Eout}{E^{\text{out}}}
\newcommand{\Eor}{E^{\text{or}}}
\newcommand{\Eand}{E^{\text{and}}}
\newcommand{\Forr}{F^{\text{or}}}
\newcommand{\Fand}{F^{\text{and}}}
\newcommand{\Kand}{K^{\text{and}}}
\newcommand{\Fout}{F^{\text{in-out}}}

Fix a monotone $(2,2)$-circuit $C$ of size $m$ together with an input vector $x$ that are provided to the data structure $\Dcirctodeg_\ell$.
By running \cref{lem:circ-to-deg-helper} on the circuit $C$ and the initial vector $x$ we can assume that: (1) every input gate of $C$ has outdegree $1$; (2) every OR-gate of $C$ has indegree~$2$ and outdegree $1$.

We describe now the graph $G_{C,x}$, and from its definition it will follow immediately how to initialize and update it within the desired time.
Since $C$ and $x$ are fixed, we write $G$ as an abbreviation for $G_{C,x}$.
Let $\orgates(C), \andgates(C)$ be the sets of all OR-gates and AND-gates in $C$, respectively.
We define the sets of vertices and edges of $G$ as partitioned into smaller sets where each smaller set corresponds to a different part of the circuit $C$:
\begin{align*}
V(G) \coloneqq & \left(\biguplus_{e \in \cwires(C)} \Vwire_{e}\right) \uplus \left(\biguplus_{p \in \cinputs(C)} \Vin_p\right) \uplus \\
& \left(\biguplus_{h \in \orgates(C)} \Vor_h\right) \uplus \left(\biguplus_{h \in \andgates(C)} \Vand_h\right) \uplus \Vout,\\
E(G) \coloneqq & \left(\biguplus_{p \in \cinputs(C)} \Ein_p\right) \uplus \left(\biguplus_{h \in \orgates(C)} \Eor_h\right) \uplus \left(\biguplus_{h \in \andgates(C)} \Eand_h\right) \uplus \Eout.
\end{align*}
These sets will have the following properties:
\begin{itemize}
\item For every $e \in \cwires(C)$, $\Vwire_e$ is an independent set consisting of $\ell-1$ vertices in~$G$.
\item For every $p \in \cinputs(C)$, $\Vin_p$ is an independent set consisting of $\ell$ vertices in~$G$.
Additionally, vertices of $\Vin_p$ can be incident only to the edges of $\Ein_p \cup \Eout$.
\item For every $h \in \orgates(C)$, $\Vor_h$ consists of $\Oh(\ell^2)$ vertices, and $\Eor_h$ is the set of all edges incident to $\Vor_h$.
\item Similarly, for every $h \in \andgates(C)$, $\Vand_h$ consists of $\Oh(\ell^2)$ vertices, and $\Eand_h$ is the set of all edges incident to $\Vand_h$.
\item The set $\Vout$ contains $\Oh(\ell^2 \cdot m)$ vertices, and $\Eout$ is the set of all edges incident to $\Vout$.
\end{itemize}

Before we delve into the formal details of each gadget, let us give a high-level overview of the reduction using the notation introduced above. The overall structure of the graph $G = G_{C,x}$ closely mirrors the topology of $C$. The core mechanism relies on how degeneracy ``propagates'' through $G$ based on the dynamic evaluation of $C(x)$:
\begin{itemize}
    \item When $C(x) = 1$, we construct a complete $\ell$-degeneracy ordering for $G$. Starting from the input wire sets $\Vwire_f$ corresponding to inputs $p$ with $x(p) = 1$, we can sequentially remove vertices that, at the moment of their removal,  have degree at most $\ell$ (that is, the consecutive gadgets \emph{$\ell$-unravel}). This unraveling propagates forward through the gadgets of gates evaluating to $1$, eventually reaching the output wire $e$ and collapsing $\Vout$. Once $\Vout$ is removed, the remaining input sets $\Vin_p$ for inputs with $x(p) = 0$ lose their edges to $\Vout$, dropping their degree to at most $\ell$, which triggers the collapse of all remaining $0$-valued gadgets and leaves $G$ completely unraveled.
    \item When $C(x) = 0$, the $0$-signal reaches the output wire $e$. By tracing the $0$-evaluations backward from the output, we can extract a ``justifying tree'' of $0$-evaluating wires and gadgets. The output gadget $\Vout$, together with the gadgets corresponding to this justifying tree, forms a self-sustaining induced subgraph $H \subseteq G$ in which every vertex maintains a degree of at least $2\ell - 1$. Thus, $\deg(G) \ge \delta(H) \ge 2\ell - 1$.
\end{itemize}
We remark that this mechanism is inspired by the proof presented in \cite[Appendix~A]{PilipczukST18a}, but the idea from there is generalized here to work for any fixed $\ell$. This is to create an approximation gap arbitrarily close to $2$.

Throughout the construction and the analysis of all the gadgets, we will maintain and explicitly rely on the following invariant for every wire vertex set:
\begin{quote}
\textbf{Property ($\star$).}
For every wire $e = (g,h) \in \cwires(C)$, the set $\Vwire_e$ consists of $\ell - 1$ vertices. Every vertex $v \in \Vwire_e$ satisfies the following:
\begin{enumerate}
    \item[(i)] $v$ is adjacent to at most $\ell$ vertices from the gadget corresponding to $g$: $\Vin_g$ (if $g$ is an input gate), $\Vor_g$ (if $g$ is an~OR-gate), or $\Vand_g$ (if $g$ is an~AND-gate); and
    \item[(ii)] $v$ is adjacent to exactly $\ell-1$ vertices from the gadget corresponding to $h$: $\Vout$ (if $h$ is the output gate), $\Vor_h$ (if $h$ is an~OR-gate) or $\Vand_h$ (if $g$ is an~AND-gate).
\end{enumerate}
In particular, this invariant guarantees that every vertex $v \in \Vwire_e$ has total degree at most $d_G(v) \leq \ell + (\ell - 1) = 2\ell - 1$.
\end{quote}

\noindent We are ready to present the detailed description of every gadget.

\paragraph*{Input gadgets.}{
Let $p \in \cinputs(C)$ be an input gate.
By the initial transformation of $C$, we know that $p$ has exactly one outgoing wire in $C$.
Denote this wire by $f$.

Recall that $\Vin_p$ is defined as an independent set of size $\ell$ in~$G$. The edge set of the input gadget is defined as follows:
\[
\Ein_p=\begin{cases}
			\Vin_p \times \Vwire_f, & \text{if $x(p)=0$,}\\
            \varnothing, & \text{if $x(p)=1$.}
		 \end{cases}
\]

Note that such a definition of the input gadget preserves Property~($\star$) for the set $\Vwire_f$ (since $|\Vin_p| = \ell$).
Additionally, every vertex $v$ of $\Vin_p$ is incident to exactly $\ell$ edges of $\Eout$ which connect $v$ with some vertices of $\Vout$.
The exact connections will be provided in the description of the output~gadget.

\begin{claim}\label{clm:input-unravel} We have the following:
    \begin{itemize}
        \item Assume that $x(p) = 1$. Then the set $\Vwire_f$ $\ell$-unravels in $G$.
        \item Let $G' \coloneqq G - \Vout$. Then the set $U \coloneqq \Vin_p \cup \Vwire_f$ $\ell$-unravels in $G'$.
    \end{itemize}
\end{claim}
\begin{claimproof}
For the first part, if $x(p) = 1$, we have $\Ein_p = \varnothing$, so there are no edges between $\Vin_p$ and $\Vwire_f$. Then, by Property~($\star$), every vertex of $\Vwire_f$ has degree $\ell-1$ in $G$, and thus the entire set $\Vwire_f$ can be $\ell$-unravelled from $G$ in an arbitrary order.

For the second part, in $G' = G - \Vout$, all the edges of $\Eout$ are removed, so every vertex $v \in \Vin_p$ is adjacent only to $\Vwire_f$, giving $d_{G'}(v) \leq |\Vwire_f| = \ell - 1$. Thus, all vertices of $\Vin_p$ can be $\ell$-unraveled first. Once $\Vin_p$ is removed, $\Vwire_f$ can be $\ell$-unraveled as in the first part of the claim.
\end{claimproof}

\begin{claim}\label{clm:input-bounds} We have the following:
    \begin{itemize}
        \item For every $v \in \Vin_p$, $d_G(v) \leq  2\ell-1$.
        \item Assume that $x(p) = 0$. Let $H \coloneqq G[\Vin_p \cup \Vout \cup \Vwire_f]$. Then, for every $v \in \Vin_p$, $d_H(v) \geq 2\ell-1$.
    \end{itemize}
\end{claim}
\begin{claimproof}
Every vertex $v \in \Vin_p$ is incident to $\ell$ edges of $\Eout$ (connecting $v$ to $\Vout$) and at most $|\Vwire_f| = \ell - 1$ edges of $\Ein_p$. Thus, $d_G(v) \le \ell + (\ell - 1) =2\ell-1$.

When $x(p) = 0$, we have $\Ein_p = \Vin_p \times \Vwire_f$. Every vertex $v \in \Vin_p$ is adjacent to all the $\ell - 1$ vertices of $\Vwire_f$ and to $\ell$ vertices in $\Vout$. Since all these neighbors lie in $V(H)$, we obtain $d_H(v) = \ell + (\ell - 1) = 2\ell - 1$.
\end{claimproof}
}

\paragraph*{OR gadgets.}{
Consider an~OR-gate $h \in \orgates(C)$.
By the initial simplification we run on $C$, we know that $h$ has two input wires $e_1, e_2$, and one output wire $f$.
For such a~gate $h$, we define the vertex set of corresponding OR-gadget:
\[
\Vor_h \coloneqq L^h_1 \uplus L^h_2 \uplus \ldots \uplus L^h_\ell
\]
as partitioned into $\ell$ layers, where each layer $L_i^h$ contains different $\ell-1$ vertices of $G$.
Note that $|\Vor_h| = \ell \cdot (\ell - 1) \leq \Oh(\ell^2)$ as desired.
The edges of the gadget are defined as follows:
\[
\Eor_h \coloneqq \left(\Vwire_{e_1} \times L_1^h\right) \cup \left(\Vwire_{e_2} \times L_{\ell}^h\right) \cup \left(\bigcup_{i=1}^{\ell-1} L_i^h \times L_{i+1}^h\right) \cup \Forr_h,
\]
where $\Forr_h \subseteq \Vor_h \times \Vwire_f$ is an arbitrary set of edges satisfying the following properties: (1) every vertex of $\Vor_h$ is incident to exactly one edge of $\Forr_h$; (2) every vertex of $\Vwire_f$ is incident to exactly $\ell$ edges of $\Forr_h$.
See \cref{fig:or-gadget} for an~example OR-gadget, and observe that $|\Eor_h| \leq \Oh(\ell^3)$.

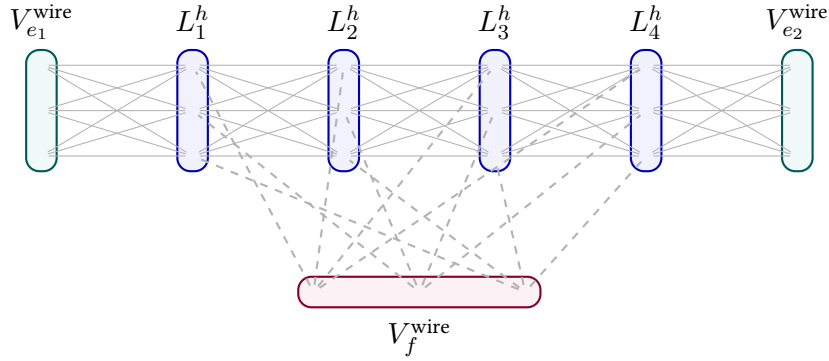
\begin{figure}[ht]
\centering
\begin{tikzpicture}[
    vertex/.style={circle, fill=black, inner sep=1.8pt},
    layerblock/.style={draw=blue!70!black, thick, rounded corners=5pt, fill=blue!5, inner sep=3pt},
    wireblock/.style={draw=teal!70!black, thick, rounded corners=5pt, fill=teal!5, inner sep=3pt},
    outblock/.style={draw=purple!70!black, thick, rounded corners=5pt, fill=purple!5, inner sep=3pt},
    std edge/.style={gray!60, thin},
    grey edge/.style={gray!60, thick, dashed}
]
    \node[vertex] (u1) at (0, 0.6) {};
    \node[vertex] (u2) at (0, 0) {};
    \node[vertex] (u3) at (0, -0.6) {};
    \node[wireblock, fit=(u1)(u3), label=above:{$\Vwire_{e_1}$}] (box_e1) {};

    \node[vertex] (w1) at (10, 0.6) {};
    \node[vertex] (w2) at (10, 0) {};
    \node[vertex] (w3) at (10, -0.6) {};
    \node[wireblock, fit=(w1)(w3), label=above:{$\Vwire_{e_2}$}] (box_e2) {};

    \node[vertex] (a1) at (2, 0.6) {};
    \node[vertex] (a2) at (2, 0) {};
    \node[vertex] (a3) at (2, -0.6) {};
    \node[layerblock, fit=(a1)(a3), label=above:{$L_1^h$}] (box_L1) {};

    \node[vertex] (b1) at (4, 0.6) {};
    \node[vertex] (b2) at (4, 0) {};
    \node[vertex] (b3) at (4, -0.6) {};
    \node[layerblock, fit=(b1)(b3), label=above:{$L_2^h$}] (box_L2) {};

    \node[vertex] (c1) at (6, 0.6) {};
    \node[vertex] (c2) at (6, 0) {};
    \node[vertex] (c3) at (6, -0.6) {};
    \node[layerblock, fit=(c1)(c3), label=above:{$L_3^h$}] (box_L3) {};

    \node[vertex] (d1) at (8, 0.6) {};
    \node[vertex] (d2) at (8, 0) {};
    \node[vertex] (d3) at (8, -0.6) {};
    \node[layerblock, fit=(d1)(d3), label=above:{$L_4^h$}] (box_L4) {};

    \node[vertex] (v1) at (3.6, -2.4) {};
    \node[vertex] (v2) at (5.0, -2.4) {};
    \node[vertex] (v3) at (6.4, -2.4) {};
    \node[outblock, fit=(v1)(v3), label=below:{$\Vwire_f$}] (box_f) {};

    \foreach \x in {u1, u2, u3} \foreach \y in {a1, a2, a3} \draw[std edge] (\x) -- (\y);
    \foreach \x in {a1, a2, a3} \foreach \y in {b1, b2, b3} \draw[std edge] (\x) -- (\y);
    \foreach \x in {b1, b2, b3} \foreach \y in {c1, c2, c3} \draw[std edge] (\x) -- (\y);
    \foreach \x in {c1, c2, c3} \foreach \y in {d1, d2, d3} \draw[std edge] (\x) -- (\y);
    \foreach \x in {w1, w2, w3} \foreach \y in {d1, d2, d3} \draw[std edge] (\x) -- (\y);

    \draw[grey edge] (v1) -- (a1);
    \draw[grey edge] (v1) -- (b1);
    \draw[grey edge] (v1) -- (c1);
    \draw[grey edge] (v1) -- (d1);

    \draw[grey edge] (v2) -- (a2);
    \draw[grey edge] (v2) -- (b2);
    \draw[grey edge] (v2) -- (c2);
    \draw[grey edge] (v2) -- (d2);

    \draw[grey edge] (v3) -- (a3);
    \draw[grey edge] (v3) -- (b3);
    \draw[grey edge] (v3) -- (c3);
    \draw[grey edge] (v3) -- (d3);
\end{tikzpicture}
\caption{OR gadget construction for $\ell=4$. Dashed edges represent an example set $\Forr_h$.}
\label{fig:or-gadget}
\end{figure}

Note that such a definition of an OR-gadget preserves Property~($\star$) for the sets $\Vwire_{e_1}$, $\Vwire_{e_2}$ and $\Vwire_{f}$.
Indeed, vertices of $\Vwire_{e_1}$ and $\Vwire_{e_2}$ are adjacent to $\ell-1$ vertices of $L_1^h$ and $L_\ell^h$, respectively, and each of the vertices of $\Vwire_f$ receives $\ell$ incident edges of $\Forr_h$.
Let us make a few more observations about OR-gadgets.

\begin{claim}\label{clm:or-unravel}
    Let $i \in \{1,2\}$ and $G' \coloneqq G - \Vwire_{e_i}$. Then the set $U \coloneqq \Vor_h \cup \Vwire_f$ $\ell$-unravels in~$G'$.
\end{claim}
\begin{claimproof}
    Without loss of generality assume $i = 1$, so $\Vwire_{e_1}$ is removed in $G'$. Consider $v \in L_1^h$. In $G'$, $v$ is adjacent only to $L_2^h$ ($|L_2^h| = \ell - 1$ edges) and $\Vwire_f$ ($1$ edge from $\Forr_h$), so $d_{G'}(v) = (\ell - 1) + 1 = \ell$. Thus, all vertices in $L_1^h$ can be pruned.
    Once $L_1^h$ is removed, any vertex in $L_2^h$ loses its $\ell - 1$ edges to $L_1^h$, leaving only edges to $L_3^h$ ($\ell - 1$ edges) and $\Vwire_f$ ($1$ edge), dropping its degree to $\ell$.
    By induction, layers $L_1^h, L_2^h, \dots, L_\ell^h$ are pruned sequentially. After all vertices of $\Vor_h$ are removed, each vertex in $\Vwire_f$ loses all $\ell$ incident edges from $\Forr_h$, leaving $\ell - 1$ edges to its target gadget (by Property~($\star$)), so $\Vwire_f$ $\ell$-unravels as well. The case $i = 2$ is symmetric, starting the unraveling sequence from $L_\ell^h$.
\end{claimproof}

\begin{claim}\label{clm:or-bounds} We have the following:
    \begin{itemize}
        \item For every $v \in \Vor_h$, $d_G(v)=2\ell-1$.
        \item Let $H \coloneqq G[\Vor_h \cup \Vwire_{e_1} \cup \Vwire_{e_2} \cup \Vwire_f]$ be an induced subgraph of $G$. Then, for every $v \in \Vor_h$, $d_H(v) = 2\ell-1$.
    \end{itemize}
\end{claim}
\begin{claimproof}
    For $v \in L_1^h$, its neighbors in $G$ lie in $\Vwire_{e_1}$ ($\ell - 1$ neighbors), $L_2^h$ ($\ell - 1$ neighbors), and $\Vwire_f$ ($1$ neighbor), so $d_G(v) = 2\ell - 1$. Similarly, for $v \in L_j^h$ ($1 < j < \ell$), its neighbors lie in $L_{j-1}^h$ ($\ell - 1$ neighbors), $L_{j+1}^h$ ($\ell - 1$ neighbors), and $\Vwire_f$ ($1$ neighbors), giving $d_G(v) = 2\ell - 1$. For $v \in L_\ell^h$, neighbors lie in $L_{\ell-1}^h$ ($\ell - 1$ neighbors), $\Vwire_{e_2}$ ($\ell - 1$ neighbors), and $\Vwire_f$ ($1$ neighbor), giving $d_G(v) = 2\ell - 1$. Thus $d_G(v) = 2\ell - 1$ for all $v \in \Vor_h$.

    Since all neighbor sets ($\Vwire_{e_1}, \Vwire_{e_2}, \Vor_h, \Vwire_f$) are contained in $V(H)$, every $v \in \Vor_h$ satisfies $d_H(v) = d_G(v) = 2\ell - 1$.
\end{claimproof}
}

\paragraph*{AND gadgets.}{
Now, let $h \in \andgates(C)$ be an~AND-gate in $C$.
Assume $h$ has two incoming wires $e_1, e_2$ and two outgoing wires $f_1, f_2$. We set $\ell_1 \coloneqq 6\ell-3$ and define:
\[
\Vand_h \coloneqq L^h_1 \uplus L^h_2 \uplus \ldots \uplus L^h_{\ell_1}
\]
as partitioned into $\ell_1$ vertex-disjoint layers, where
\[
|L_i| = \begin{cases}
\ell, & \text{if $i \not\equiv 1 \pmod{3}$},\\
            \ell-1, & \text{if $i \equiv 1 \pmod{3}$}.
\end{cases}
\]
Note that $|\Vand_h| \leq \ell_1 \cdot \ell \leq \Oh(\ell^2)$ as desired.
Let $\mathcal{L}_h \coloneqq \bigcup_{j=1}^{2\ell-2} L^h_{3j}$. We define the edge set of the gadget as follows:
\[
\Eand_h \coloneqq \left(\Vwire_{e_1} \times L_1\right) \cup \left(\Vwire_{e_2} \times L_1\right) \cup \left(\bigcup_{i=1}^{\ell_1-1} L_i \times L_{i+1}\right) \cup \Fand_h \cup \Kand_h
\]
where $\Kand_h$ is the complete set of edges on $L^h_{\ell_1}$, and $\Fand_h \subseteq \mathcal{L}_h \times (\Vwire_{f_1} \cup \Vwire_{f_2})$ is any set of edges satisfying: (1) every vertex of $\mathcal{L}_h$ is incident to at most one edge of $\Fand_h$; (2) every vertex of $\Vwire_{f_1} \cup \Vwire_{f_2}$ is incident to exactly $\ell$ edges of $\Fand_h$.
See \cref{fig:and-gadget} for an example AND-gadget, and observe that $|\Eand_h| \leq \Oh(\ell^3)$.

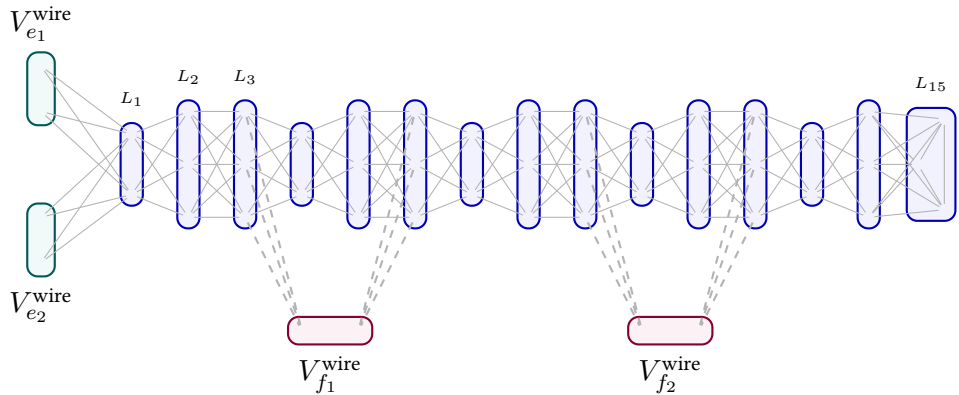
\begin{figure}[ht]
\centering
\begin{tikzpicture}[
    vertex/.style={circle, fill=black, inner sep=1.4pt},
    layerblock/.style={draw=blue!70!black, thick, rounded corners=4pt, fill=blue!5, inner sep=2pt},
    wireblock/.style={draw=teal!70!black, thick, rounded corners=4pt, fill=teal!5, inner sep=3pt},
    outblock/.style={draw=purple!70!black, thick, rounded corners=4pt, fill=purple!5, inner sep=3pt},
    std edge/.style={gray!60, thin},
    grey edge/.style={gray!60, thick, dashed}
]
    \node[vertex] (u1) at (-1.2, 1.3) {};
    \node[vertex] (u2) at (-1.2, 0.7) {};
    \node[wireblock, fit=(u1)(u2), label=above:{$\Vwire_{e_1}$}] (box_e1) {};

    \node[vertex] (w1) at (-1.2, -0.7) {};
    \node[vertex] (w2) at (-1.2, -1.3) {};
    \node[wireblock, fit=(w1)(w2), label=below:{$\Vwire_{e_2}$}] (box_e2) {};

    \foreach \i in {1,...,14} {
        \pgfmathsetmacro{\xpos}{(\i-1)*0.75}
        \pgfmathparse{mod(\i, 3) == 1 ? 1 : 0}
        \ifnum\pgfmathresult=1
            \node[vertex] (v_\i_1) at (\xpos, 0.4) {};
            \node[vertex] (v_\i_2) at (\xpos, -0.4) {};
            \node[layerblock, fit=(v_\i_1)(v_\i_2)] (box_L_\i) {};
        \else
            \node[vertex] (v_\i_1) at (\xpos, 0.7) {};
            \node[vertex] (v_\i_2) at (\xpos, 0.0) {};
            \node[vertex] (v_\i_3) at (\xpos, -0.7) {};
            \node[layerblock, fit=(v_\i_1)(v_\i_2)(v_\i_3)] (box_L_\i) {};
        \fi
    }

    \pgfmathsetmacro{\xpos}{14*0.75} 
    \node[vertex] (v_15_1) at (\xpos+0.25, 0.6) {};
    \node[vertex] (v_15_2) at (\xpos-0.10, 0.0) {};
    \node[vertex] (v_15_3) at (\xpos+0.25, -0.6) {};
    \node[layerblock, fit=(v_15_1)(v_15_2)(v_15_3)] (box_L_15) {};

    \node[above=2pt of box_L_1] {\tiny $L_1$};
    \node[above=2pt of box_L_2] {\tiny $L_2$};
    \node[above=2pt of box_L_3] {\tiny $L_3$};
    \node[above=2pt of box_L_15] {\tiny $L_{15}$};

    \foreach \x in {u1, u2, w1, w2} {
        \foreach \y in {1, 2} {
            \draw[std edge] (\x) -- (v_1_\y);
        }
    }

    \foreach \i [evaluate=\i as \nexti using int(\i+1)] in {1,...,14} {
        \pgfmathparse{mod(\i, 3) == 1 ? 2 : 3}
        \edef\sizeA{\pgfmathresult}
        \pgfmathparse{mod(\nexti, 3) == 1 ? 2 : 3}
        \edef\sizeB{\pgfmathresult}

        \foreach \a in {1,...,\sizeA} {
            \foreach \b in {1,...,\sizeB} {
                \draw[std edge] (v_\i_\a) -- (v_\nexti_\b);
            }
        }
    }

    \draw[std edge] (v_15_1) -- (v_15_2);
    \draw[std edge] (v_15_2) -- (v_15_3);
    \draw[std edge] (v_15_1) -- (v_15_3);

    \node[vertex] (f1_1) at (2.25, -2.2) {}; 
    \node[vertex] (f1_2) at (3.00, -2.2) {}; 
    \node[outblock, fit=(f1_1)(f1_2), label=below:{$\Vwire_{f_1}$}] (box_f1) {};

    \node[vertex] (f2_1) at (6.75, -2.2) {}; 
    \node[vertex] (f2_2) at (7.50, -2.2) {}; 
    \node[outblock, fit=(f2_1)(f2_2), label=below:{$\Vwire_{f_2}$}] (box_f2) {};

    \draw[grey edge] (v_3_1) -- (f1_1);
    \draw[grey edge] (v_3_2) -- (f1_1);
    \draw[grey edge] (v_3_3) -- (f1_1);

    \draw[grey edge] (v_6_1) -- (f1_2);
    \draw[grey edge] (v_6_2) -- (f1_2);
    \draw[grey edge] (v_6_3) -- (f1_2);

    \draw[grey edge] (v_9_1) -- (f2_1);
    \draw[grey edge] (v_9_2) -- (f2_1);
    \draw[grey edge] (v_9_3) -- (f2_1);

    \draw[grey edge] (v_12_1) -- (f2_2);
    \draw[grey edge] (v_12_2) -- (f2_2);
    \draw[grey edge] (v_12_3) -- (f2_2);
\end{tikzpicture}
\caption{AND gadget construction for $\ell=3$ ($\ell_1 = 15$). Dashed edges represent an example set~$\Fand_h$. The vertices of $L_{\ell_1}=L_{15}$ form a clique of size $\ell = 3$.}
\label{fig:and-gadget}
\end{figure}

Again, observe that Property~($\star$) is preserved by the AND-gadget: vertices of $\Vwire_{e_1}$ and $\Vwire_{e_2}$ are adjacent to $\ell-1$ vertices of $L_1^h$, and vertices of $\Vwire_f$ receive $\ell$ incident edges of $\Fand_h$.
Let us show a few more properties of AND-gadgets, analogously as we did for OR-gadgets.

\begin{claim}\label{clm:and-unravel}
    Let $G' \coloneqq G - (\Vwire_{e_1} \cup \Vwire_{e_2})$. Then the set $U \coloneqq \Vand_h \cup \Vwire_{f_1} \cup \Vwire_{f_2}$ $\ell$-unravels in $G'$.
\end{claim}
\begin{claimproof}
    In $G'$, vertices in $L_1^h$ have lost all edges to $\Vwire_{e_1}$ and $\Vwire_{e_2}$. They are adjacent only to $L_2^h$ ($|L_2^h| = \ell$), so $d_{G'}(v) = \ell$ for all $v \in L_1^h$. Thus $L_1^h$ can be unraveled.
    Once $L_1^h$ is removed, vertices in $L_2^h$ lose their $|L_1^h| = \ell-1$ edges, leaving edges to $L_3^h$ ($|L_3^h| = \ell$), so their degree becomes $\ell$, allowing $L_2^h$ to be unraveled.
    Next, vertices in $L_3^h$ lose their edges to $L_2^h$, leaving $\ell-1$ edges to $L_4^h$ and single edges of $\Fand_h$, giving degree $\ell$ in total, so $L_3^h$ can be unraveled.
    Inductively, all layers $L_1^h, L_2^h, \dots, L_{\ell_1-1}^h$ are unraveled sequentially.
    Then, vertices of $L_{\ell_1}$ have degree $\ell-1$ each, so the set $L_{\ell_1}$ can be removed from the graph.
    Finally, once $\Vand_h$ is removed, every vertex in $\Vwire_{f_1} \cup \Vwire_{f_2}$ loses its $\ell$ incident edges from $\Fand_h$, leaving at most $\ell - 1$ edges to its target gadget (by Property~($\star$)), so $\Vwire_{f_1} \cup \Vwire_{f_2}$ $\ell$-unravels as well.
\end{claimproof}

\begin{claim}\label{clm:and-bounds} We have the following:
    \begin{itemize}
        \item For every $v \in \Vand_h$, $d_G(v) < \maxdeg$.
        \item Let $i \in \{1,2\}$ and $H \coloneqq G[\Vand_h \cup \Vwire_{e_i}]$ be an induced subgraph of $G$. Then, for every $v \in \Vand_h$, $d_H(v) \geq 2\ell-1$.
    \end{itemize}
\end{claim}
\begin{claimproof}
    For the upper bound, a vertex in $L_1^h$ has degree at most $|\Vwire_{e_1}| + |\Vwire_{e_2}| + |L_2^h| = (\ell-1) + (\ell-1) + \ell = 3\ell - 2$. A vertex in $L_{\ell_1}^h$ has degree $|L_{\ell_1-1}^h| + (|L_{\ell_1}^h| - 1) = \ell + (\ell - 1) = 2\ell-1$. All other layer degrees are at most $2\ell + 1$. Thus $d_G(v) < \maxdeg$ for all $v \in \Vand_h$.

    For the lower bound in $H$, even if only one input wire $e_i$ is present in $H$, any $v \in L_1^h$ is incident to $|\Vwire_{e_i}| = \ell-1$ edges from $\Vwire_{e_i}$ and $|L_2^h| = \ell$ edges from $L_2^h$, giving $d_H(v) = 2\ell - 1$.
    For $2 \leq j \leq \ell_1-1$, $v \in L_j^h$ receives at least $2\ell - 1$ edges internally from $L^h_{j-1} \cup L^h_{j+1}$, which are all present in $H$.
    Finally, vertices of $L_\ell$ receive $\ell$ edges from $L_{\ell_1-1}$ and $\ell-1$ edges within the layer $L_{\ell_1}$, summing up to $2\ell-1$ neighbors in $H$.
    Hence $d_H(v) \ge 2\ell - 1$ for all $v \in \Vand_h$.
\end{claimproof}

In the description of the~AND-gadget, we assumed that $h$ has two input wires and two output wires.
Let us discuss briefly what changes when this is not the case.
\begin{itemize}
    \item If $h$ has only one outgoing wire $f_1$ we remove $\Vwire_{f_2}$ with adjacent edges from the gadget description: some vertices of $\Vand_h$ will be adjacent to the edges of $\Fand_h$, and some not.
    One can verify that \cref{clm:and-unravel,clm:and-bounds} do not require both $\Vwire_{f_1}$ and $\Vwire_{f_2}$ to be present.
    \item If $h$ has only one incoming wire $e_1$, we omit the set $\Vwire_{e_2}$ from the description of the gadget.
    Again, \cref{clm:and-unravel,clm:and-bounds} still hold for such a simplified gadget.
\end{itemize}
}

\paragraph*{Output gadget.}{
Let $e$ be the output wire of the circuit $C$ and let $r \coloneqq |\cinputs(C)|$ be the total number of inputs to the circuit.
We set $\ell_2 \coloneqq 3\ell r + 3$ and define the vertex set of the output gadget to be
\[
\Vout \coloneqq L^{\text{out}}_1 \uplus L^{\text{out}}_2 \uplus \ldots \uplus L^{\text{out}}_{\ell_2},
\]
as partitioned into $\ell_2$ vertex-disjoint layers whose sizes strictly follow the periodic pattern used in the AND-gadget:
\[
|L^{\text{out}}_i| = \begin{cases}
\ell, & \text{if } i \not\equiv 1 \pmod{3},\\
\ell-1, & \text{if } i \equiv 1 \pmod{3}.
\end{cases}
\]
Note that $|\Vout| \leq \ell_2 \cdot \ell \leq \Oh(\ell^2 r) \leq \Oh(\ell^2 m)$ as desired.
Let $\mathcal{L}_{\text{out}} \coloneqq \bigcup_{j=1}^{\ell r} L^{\text{out}}_{3j}$. We define the edges of the gadget as:
\[
\Eout \coloneqq \left( \Vwire_e \times L^{\text{out}}_1 \right) \cup \left( \bigcup_{i=1}^{\ell_2-1} L^{\text{out}}_i \times L^{\text{out}}_{i+1} \right) \cup \Fout \cup \Kand_{\text{out}}
\]
where $\Kand_{\text{out}}$ is the complete set of edges on $L^{\text{out}}_{\ell_2}$, and $\Fout \subseteq \mathcal{L}_{\text{out}} \times (\bigcup_{p \in \cinputs(C)} \Vin_p)$ is any set of edges satisfying: (1) every vertex of $\mathcal{L}_{\text{out}}$ is incident to exactly one edge of $\Fout$; (2) every vertex of $\Vin_p$ for every input gate $p \in \cinputs(C)$ is incident to exactly $\ell$ edges of $\Fout$.

Observe that Property~($\star$) is preserved by the output gadget: vertices of $\Vwire_e$ are adjacent to $\ell-1$ vertices of $L_1^{\text{out}}$. 
Moreover, the total number of edges required to connect all the vertices of all the sets $\Vin_p$ via $\Fout$ is exactly $r \cdot (\ell \cdot \ell) = \ell^2 r$, which matches exactly the number of vertices in $\mathcal{L}_{\text{out}}$ (since $|\mathcal{L}_{\text{out}}| = \ell r \cdot \ell = \ell^2 r$).
Let us show a few more properties of the output gadget, analogously as we did for AND-gadgets.

\begin{claim}\label{clm:out-unravel}
    Let $G' \coloneqq G - \Vwire_e$. Then the set $\Vout$ $\ell$-unravels in $G'$.
\end{claim}
\begin{claimproof}
    In $G'$, vertices in $L_1^{\text{out}}$ have lost all edges to $\Vwire_e$. They are adjacent only to $L_2^{\text{out}}$ ($|L_2^{\text{out}}| = \ell$), so $d_{G'}(v) = \ell$ for all $v \in L_1^{\text{out}}$. Thus $L_1^{\text{out}}$ can be pruned.
    Once $L_1^{\text{out}}$ is removed, every vertex of $L_2^{\text{out}}$ loses $|L_1^{\text{out}}| = \ell-1$ incident edges, leaving only edges to $L_3^{\text{out}}$ ($|L_3^{\text{out}}| = \ell$). So their degree drops to $\ell$, allowing $L_2^{\text{out}}$ to be pruned.
    Next, vertices in $L_3^{\text{out}}$ lose their edges to $L_2^{\text{out}}$, leaving edges to $L_4^{\text{out}}$ ($|L_4^{\text{out}}| = \ell-1$) and $1$ edge in $\Fout$, giving a total degree of at most $(\ell-1)+1=\ell$, so $L_3^{\text{out}}$ can be pruned.
    Inductively, all layers $L_1^{\text{out}}, L_2^{\text{out}}, \dots, L_{\ell_2-1}^{\text{out}}$ are pruned sequentially.
    Then, vertices of $L_{\ell_2}^{\text{out}}$ have degree $\ell-1$ each, so the set $L_{\ell_2}^{\text{out}}$ can be removed from the graph.
\end{claimproof}

\begin{claim}\label{clm:out-bounds} We have the following:
    \begin{itemize}
        \item For every $v \in \Vout$, $d_G(v) < \maxdeg$.
        \item Let $H \coloneqq G[\Vout \cup \Vwire_e]$ be an induced subgraph of $G$. Then, for every $v \in \Vout$, $d_H(v) \geq 2\ell-1$.
    \end{itemize}
\end{claim}
\begin{claimproof}
    For the upper bound, a vertex in $L_1^{\text{out}}$ has degree at most $|\Vwire_e| + |L_2^{\text{out}}| = (\ell-1) + \ell = 2\ell-1$. A vertex in $L_{\ell_2}^{\text{out}}$ has degree $|L_{\ell_2-1}^{\text{out}}| + (|L_{\ell_2}^{\text{out}}| - 1) = \ell + (\ell - 1) = 2\ell-1$. All other layer degrees are at most $2\ell + 1 < \maxdeg$. Thus $d_G(v) < \maxdeg$ for all $v \in \Vout$.

    For the lower bound in $H$, assuming the incoming wire $e$ is present in $H$, every vertex $v \in L_1^{\text{out}}$ retains $|\Vwire_e| = \ell - 1$ edges to the incoming wire and $|L_2^{\text{out}}| = \ell$ edges to $L_2^{\text{out}}$, giving $d_H(v) = 2\ell - 1$. For $2 \le i \le \ell_2-1$, $v \in L_i^{\text{out}}$ receives at least $2\ell - 1$ edges internally from $L^{\text{out}}_{i-1} \cup L^{\text{out}}_{i+1}$, which are all present in $H$.
    Finally, vertices of $L_{\ell_2}^{\text{out}}$ receive $\ell$ edges from $L_{\ell_2-1}^{\text{out}}$ and $\ell-1$ edges within the layer $L_{\ell_2}^{\text{out}}$, summing up to $2\ell-1$ neighbors in $H$.
    Hence $d_H(v) \ge 2\ell - 1$ for all $v \in \Vout$.
\end{claimproof}
}

This completes the construction of the graph $G_{C,x}$. We now analyze the necessary properties.

\paragraph*{Graph size.}{
Let $m$ denote the size of the circuit $C$ (the total number of gates and wires). We now bound the size of the constructed graph $G_{C,x}$. The graph contains $\Oh(m)$ wire sets, each of size $\ell-1$, contributing $\Oh(m\ell)$ vertices. Each OR-gadget contains $\Oh(\ell^2)$ vertices, and each AND-gadget consists of $6\ell-3$ layers of size at most $\ell$, contributing $\Oh(\ell^2)$ vertices per gate. The output gadget aggregates back-edges to $r \le m$ input gadgets and contains $\ell_2 = 3\ell r + 3$ layers, adding $\Oh(\ell^2 r) \leq \Oh(\ell^2 m)$ vertices. Summing these up, the total number of vertices is $|V(G_{C,x})| \leq \Oh(\ell^2 m)$. 

By Property~($\star$) and \cref{clm:input-bounds,clm:or-bounds,clm:and-bounds,clm:out-bounds} every vertex of $G$ has degree less than $\maxdeg$. Consequently, the total number of edges is $|E(G_{C,x})| \leq \Oh(\ell^3 m)$.
}

\paragraph*{Initialization and updates.}{Given a circuit $C$ and the initial input vector $x_0$ the graph $G_{C,x_0}$ can be constructed layer by layer (following topological ordering of $C$) in time proportional to the graph size, which is $\Oh(\ell^3 m)$.

When the input assignment $x$ changes (e.g., a single variable $p$ is flipped), we only update the corresponding input gadget in $G_{C,x}$. This local update consists exclusively of modifying the edge set $\Ein_p$: inserting or removing exactly $\ell(\ell-1)$ edges between $\Vin_p$ and $\Vwire_f$. Crucially, the internal structure of the logic gates, wires, and output gadget remains fully static; hence, updates take strictly $\Oh(\ell^2)$ time per modified input.
}

\paragraph*{Circuit evaluation vs degeneracy.}{
It remains to analyze the relationship is between the evaluation of circuit $C$ on $x$ and the degeneracy of the graph $G = G_{C,x}$.
\cref{clm:eval-one-low-deg,clm:eval-zero-high-deg} below establish the desired inequalities.
\begin{claim}\label{clm:eval-one-low-deg}
    Assume that $C(x) = 1$. Then, $\deg(G) \leq \ell$.
\end{claim}
\begin{claimproof}
    We construct a valid $\ell$-degeneracy ordering for $G$, demonstrating that the graph can be reduced to an empty graph by $\ell$-unraveling, that is, sequentially removing vertices of current degree at most $\ell$. The unraveling proceeds inductively from the input gadgets to the output gadget. 

    Since $C(x) = 1$, the Boolean evaluation assigns $1$ to a set of wires and gates propagating from the inputs to the output.
    For every input $p$ with $x(p) = 1$, the corresponding set $\Vwire_f$ can be unraveled by \cref{clm:input-unravel}. 
    For any OR-gate evaluating to $1$, at least one input wire evaluates to $1$ and thus is unraveled, which drops the degrees in the OR-gadget and allows its output wire gadget to be unraveled by \cref{clm:or-unravel}. 
    For any AND-gate evaluating to $1$, both input wires evaluate to $1$ and are unraveled, triggering the sequential removal of its output wires by \cref{clm:and-unravel}. 
    Because $C(x) = 1$, the output wire $e$ evaluates to $1$ and is unraveled. Without the support of this incoming wire, the output gadget $\Vout$ can be unraveled layer by layer by \cref{clm:out-unravel}.

    Once $\Vout$ is completely removed, the remaining graph $G' = G - \Vout$ collapses. For every input $p$ with outgoing wire $f$ such that $x(p) = 0$, the set $\Vin_p \cup \Vwire_f$ can be unraveled in $G'$ by \cref{clm:input-unravel}. This recursively drops the incoming degrees for all remaining $0$-evaluating OR and AND gadgets, allowing them (and their outgoing wires) to be sequentially unraveled. Consequently, all vertices in $G$ are successfully eliminated, proving that $\deg(G) \le \ell$.
\end{claimproof}

\begin{claim}\label{clm:eval-zero-high-deg}
    Assume that $C(x) = 0$. Then, $2\ell -1 \leq \deg(G) \leq \maxdeg$.
\end{claim}
\begin{claimproof}
    The upper bound $\deg(G) \le \maxdeg$ follows directly from the degree bounds proved in Property~($\star$) and \cref{clm:input-bounds,clm:or-bounds,clm:and-bounds,clm:out-bounds}.
    For the lower bound, it suffices to exhibit an induced subgraph $H \subseteq G$ with minimum degree $\delta(H) \ge 2\ell-1$. We construct $H$ by tracing a ``justifying tree'' of $0$-evaluations backward from the output gate. We let $H \coloneqq G[U]$, and define the set $U \subseteq V(G)$ by following the rules below recursively:
    \begin{enumerate}
        \item Include the entire output gadget $\Vout$ in $U$.
        \item Since $C(x)=0$, the output wire $f$ evaluates to $0$. Include $\Vwire_f$ in $U$.
        \item For any set $\Vwire_e$ already included in $U$, let $h$ be its source gate. Include the gadget of $h$ in $U$ ($\Vor_h$, $\Vand_h$, or $\Vin_h$).
        \item Depending on the type of the source gate $h$ (which must evaluate to $0$ since its outgoing wire $e$ evaluates to $0$):
        \begin{itemize}
            \item If $h$ is an OR-gate, both of its input wires $e_1$ and $e_2$ evaluate to $0$. Include both sets $\Vwire_{e_1}$ and $\Vwire_{e_2}$ in $U$.
            \item If $h$ is an AND-gate, at least one of its input wires $e$ evaluates to $0$. Pick exactly one such $0$-evaluating input wire and include $\Vwire_e$ in $U$.
            \item If $h$ is an input gate $p$, include the set $\Vin_p$ in $U$. Note that $x(p)=0$ must hold as $e$ evaluates to $0$.
        \end{itemize}
    \end{enumerate}
    Let $H = G[U]$ be the subgraph of $G$ induced by all the vertices chosen in this process. We now verify that $\delta(H) \ge 2\ell-1$.

    By construction, every wire $e$ included in $H$ has both its source gadget and its target gadget in~$H$. By Property ($\star$), its vertices maintain exactly $\ell$ edges from the source and $\ell-1$ edges to the target, giving them a degree of $2\ell-1$ in $H$.
    For every included OR-gate evaluating to $0$, both of its incoming wires are in $H$, satisfying the conditions of \cref{clm:or-bounds}.
    For every included AND-gate evaluating to~$0$, exactly one of its incoming wires is in $H$, which by \cref{clm:and-bounds} is sufficient to maintain a degree of at least $2\ell-1$.
    For every included input gate $p$, $x(p)=0$ holds. The vertices of $\Vin_p$ maintain $\ell-1$ edges to their outgoing wire set and $\ell$ edges to $\Vout$ (since $\Vout \subseteq H$), guaranteeing a degree of $2\ell-1$ by \cref{clm:input-bounds}.
    Finally, if $f$ is the output wire, the set $\Vwire_f$ is in $H$, which is sufficient for every vertex in $\Vout$ to maintain a degree of at least $2\ell-1$ by \cref{clm:out-bounds}. 
    So $2\ell-1 \le \delta(H) \le \deg(G)$.
\end{claimproof}
}

We have proved that if $C(x) = 1$ then $\deg(G) \leq \ell$, and if $C(x) = 0$ then $2\ell-1 \leq \deg(G) \leq \maxdeg$.
However, for $\ell \geq 2$, the intervals $[0, \ell]$ and $[2\ell-1, \maxdeg]$ are disjoint, hence if $\deg(G) \leq \ell$ then it must hold that $C(x) = 1$, and similarly if $2\ell-1 \leq \deg(G) \leq \maxdeg$ then $C(x) = 0$.
}
\end{proof}


Having proved \cref{lem:circ-to-deg}, we are ready to conclude the main result of this work.

\begin{proof}[Proof of \cref{thm:main}]
Suppose that for some $\eps, \delta > 0$, $D, k \in \N$, and a function $f \colon \N \to \N$, the problem of \dyndeg{\eps}{D} admits a data structure $\Ddeg$ with initialization time $f(d_\mx)\cdot n^k$ and amortized update time $f(d_\mx)\cdot n^{1-\delta}$ where $n$ is the number of vertices of the input graph. We may assume without loss of generality that $f$ is non-decreasing.

Using $\Ddeg$ as an auxiliary data structure, we show that $\dyncir$ restricted to monotone $(2,2)$-circuits admits a data structure $\Dcirc$ with initialization time $\Oh(m^{k'})$ and amortized update time $\Oh(m^{1-\delta})$, where $k' \coloneqq \max(k, 2)$ and $m$ is the size of the input circuit.
This will contradict \cref{lem:circuit-lower-bound}.

Let $\ell \coloneqq 1 + \lceil \frac{D+1}{\eps} \rceil \geq 2$ and $d_\mx \coloneqq \maxdeg$. Note that $\ell, d_\mx \leq \Oh(1)$ as $D$ and $\eps$ are universal constants here.
During the run, the data structure $\Dcirc$ maintains:
\begin{itemize}
\item a single instance $\Dcirctodeg_\ell$ of the data structure from \cref{lem:circ-to-deg},
\item a single instance $\Ddeg$ of the data structure from the assumption.
\end{itemize}

Given a monotone Boolean $(2,2)$-circuit $C$ and a dynamic input vector $x$ we are going to plug the pair $(C, x)$ into the data structure $\Dcirctodeg_\ell$.
The data structure $\Dcirctodeg_\ell$ returns a corresponding graph $G_{C, x}$ which is going to serve as an input (together with the value of $d_\mx$) to the data structure $\Ddeg$.
Observe that \cref{lem:circ-to-deg} guarantees that $\deg(G_{C,x}) \leq \maxdeg = d_\mx$ at all times.

First, let us show that using these data structures as presented, we can efficiently evaluate the circuit $C$ on the vector~$x$.

\begin{claim}
After initialization and every update, the data structure $\Dcirc$ can compute the value of $C(x)$ in time $\Oh(1)$.
\end{claim}
\begin{claimproof}
Let $G_{C, x}$ be the current graph returned by the data structure $\Dcirctodeg_\ell$.
By the definition, the data structure $\Ddeg$ runs on the graph $G_{C, x}$ and maintains an integer $d$ satisfying the condition
\[
    \deg(G_{C, x}) \leq d \leq (2-\eps) \deg(G_{C, x}) + D.
\]
It is enough to show that based on the value of $d$, we can deduce whether $C(x) = 0$ or $C(x) = 1$.
By \cref{lem:circ-to-deg} we know that $C(x) = 0$ is equivalent to $\deg(G_{C, x}) \geq 2\ell - 1$, and then $d \geq \deg(G_{C, x}) \geq 2\ell - 1$. 
On the other hand, $C(x) = 1$ is equivalent to $\deg(G_{C, x}) \leq \ell$, and then
\[
    d \leq (2-\eps) \deg(G_{C, x}) + D \leq (2-\eps) \cdot \ell + D < \left(2 - \frac{D+1}{\ell}\right) \cdot \ell + D = 2 \ell - 1.
\]
Hence, by comparing the values of $d$ and $2 \ell - 1$ we can evaluate $C(x)$.
\end{claimproof}

It remains to be discussed what is the total time needed to initialize and update the data structures $\Dcirctodeg_\ell$ and $\Ddeg$.

\paragraph*{Initialization.}{
The data structure $\Dcirc$ receives as input a monotone Boolean circuit $C$ and an initial input vector $x_0$.
Let $m$ be the number of gates in the circuit $C$.
First, we initialize our inner structure $\Dcirctodeg_\ell$ with the pair $(C, x_0)$; this returns the graph $G_{C, x_0}$.
By \cref{lem:circ-to-deg} this takes time $\Oh(\ell^3 \cdot m) \leq \Oh_{\eps,D}(m)$.
Let $n = |V(G_{C, x_0})| \leq \Oh(\ell^2 \cdot m)$, and recall that $|E(G_{C, x_0})| \leq \Oh(\ell^3 \cdot m)$.

Next, we need to provide the graph $G_{C,x}$ to the structure $\Ddeg$.
We achieve this in two steps. First, initialize the structure $\Ddeg$ with an empty graph on $n$ vertices (together with the value of $d_\mx$) which takes time
\[
f(d_\mx) \cdot n^k \leq f(\maxdeg) \cdot \Oh\left((\ell^2 \cdot m)^k\right) \leq \Oh_{\eps,D,k}\left(m^k\right).
\]
Second, provide edges of $E(G_{C, x_0})$ as a sequence of $\Oh(\ell^3 \cdot m)$ edge updates to $\Ddeg$. This takes time
\[
\Oh(\ell^3 \cdot m) \cdot f(d_\mx) \cdot n^{1-\delta} = \Oh(\ell^3 \cdot m) \cdot f(\maxdeg) \cdot \Oh\left((\ell^2 \cdot m)^{1-\delta}\right) \leq \Oh_{\eps, D, \delta}\left(m^{2-\delta}\right).
\]
Hence, the total initialization time can be upper bounded by $\Oh_{\eps,D,k,\delta}(m^k + m^{2-\delta}) \leq \Oh(m^{k'})$, since $\eps$, $D$, $k$ and $\delta$ are universal constants.
}
\paragraph*{Update.}{
Consider an update of the input vector $x$ which is flipping the $i$-th bit of $x$ for some $i \in \N$, and let $x'$ denote the new input vector.
First, we need to run $\oper{flip}(i)$ on the data structure $\Dcirctodeg_\ell$ which returns the graph $G_{C,x'}$ as a sequence of $\Oh(\ell^2)$ edge updates to $G_{C,x}$.
This step runs in time $\Oh(\ell^2) \leq \Oh_{\eps,D}(1)$.
Next, all these $\Oh(\ell^2)$ edge updates need to be relayed to the data structure $\Ddeg$, which takes total amortized time
\[
\Oh(\ell^2) \cdot f(d_\mx) \cdot n^{1-\delta} = \Oh(\ell^2) \cdot f(\maxdeg) \cdot \Oh\left((\ell^2 \cdot m)^{1-\delta}\right) \leq \Oh_{\eps,D,\delta}\left(m^{1-\delta}\right).
\]
Therefore, the total amortized update time is $\Oh_{\eps,D,\delta}(m^{1-\delta}) = \Oh(m^{1-\delta})$, as $\eps$, $D$ and $\delta$ are universal constants.
}

Summing up, we proved that the assumed data structure $\Dcirc$ can be used to solve the problem $\dyncir$ restricted to monotone $(2,2)$-circuits with initialization time $\Oh(m^{k'})$ and amortized update time $\Oh(m^{1-\delta})$.
This is a contradiction with \cref{lem:circuit-lower-bound}, and ends the proof of our main theorem.
\end{proof}

\section*{AI disclosure}
\label{sec:ai-disclosure}

During the preparation of this work, the authors used Gemini 3.6 Flash to refine the language and style, assist with proof verification, and generate LaTeX code for \cref{fig:example-circuit,fig:example-ov-to-circ,fig:or-gadget,fig:and-gadget}. The theoretical concepts and core proofs were entirely developed by the authors. All outputs were thoroughly reviewed, checked, and edited by the authors, who take full responsibility for the contents of this paper.

\bibliography{auxiliary/references}

\end{document}